\documentclass[manuscript,screen]{acmart}

\usepackage{graphicx}

\usepackage{soul,xcolor}

\setstcolor{red}

\usepackage{xcolor}

\usepackage{amsmath}

\usepackage{amsfonts}

\usepackage{wrapfig}

\usepackage[linesnumbered,ruled,vlined]{algorithm2e}

\newcommand{\tnm}[1]{{\textbf{\texttt{#1}}}} 

\newcommand{\sphere}{\mathcal{S}}

\SetCommentSty{mycommfont}

\SetKwInput{KwInput}{Input}                
\SetKwInput{KwOutput}{Output}              

\newcommand{\deq}{\mathrel{\stackrel{\scriptscriptstyle\Delta}{=}}}

\usepackage{caption}
\AtBeginDocument{%
  \providecommand\BibTeX{{%
    \normalfont B\kern-0.5em{\scshape i\kern-0.25em b}\kern-0.8em\TeX}}}

\acmBooktitle{}

\begin{document}

\title{Robustness of Safety for Linear Dynamical Systems: Symbolic and Numerical Approaches}



\author{Bineet Ghosh}
\email{bineet@cs.unc.edu}
\orcid{0000-0002-1371-2803}
\affiliation{%
  \institution{University of North Carolina at Chapel Hill}
  \city{Chapel Hill}
  \state{North Carolina}
  \country{USA}
}

\author{Parasara Sridhar Duggirala}
\email{psd@cs.unc.edu}
\affiliation{%
  \institution{University of North Carolina at Chapel Hill}
  \city{Chapel Hill}
  \state{North Carolina}
  \country{USA}
}

\begin{abstract}
  In this paper, we study the robustness of safety properties of a linear dynamical system with respect to model uncertainties. Our paper involves three parts. In the first part, we provide symbolic (analytical) and numerical (representation based) techniques for computing the reachable set of uncertain linear systems. The analytical techniques extend the classical methods from perturbation theory to interval matrices. In the representation based technique, the effect of uncertainty is computed separately and added (Minkowski summation) with the reachable set of the nominal system. In the second part, we present the relationship between the singular values of the dynamics and the reachable set. Leveraging the existing literature, we determine the cells in the system dynamics that are the most and least sensitive to perturbation. We further prove a relationship between the reachable set of a linear uncertain system and the maximum singular value of the uncertain dynamics matrix. Finally, we propose two heuristics to compute the robustness threshold of the system --- the maximum uncertainty that can be introduced to the system without violating the safety property. We evaluate the reachable set computation techniques, effects of singular values, and estimation of robustness threshold on two case studies from varied domains, illustrating the applicability, practicality and scalability of the artifacts, proposed in this paper, on real-world examples. We further evaluate our artifacts on several linear dynamical system benchmarks. To the best of the authors' knowledge, this is the first work to: (i) extend perturbation theory to compute reachable sets of linear uncertain systems, (ii) leverage the relationship between the reachable set of a linear system and the maximum singular values to determine the effect of uncertainties and (3) estimate the threshold of robustness that can be tolerated by the system while remaining safe.
  
\end{abstract}



\maketitle

\section{Introduction}
Safety analysis of cyber-physical systems require a precise mathematical model of the system. 
Since constructing a precise model is subjected to measurement and sensor uncertainties, it is important to understand the robustness of the safety analysis, \textit{i.e.}, whether the safety specification holds true for some perturbations in the model.
For safety critical systems, understanding robustness of safety analysis is even more essential.

A (continuous time) linear dynamical system without uncertainties is defined as $\dot{x}=Ax$, where the the dynamics matrix $A$ is assumed to be known accurately, $i.e.$, all the elements of the matrix are known accurately. 
Consider the following Example \ref{ex:linearSys} of a linear system.
\begin{example}
\label{ex:linearSys}
Consider the continuous time linear dynamical system $\dot{x} = Ax$, where $A$ is 
$
\begin{bmatrix}
    1  & c \\
    0  & 2 
\end{bmatrix}
$
and $c$ represents a system parameter.
\end{example}
Traditional methods of reachable set computation of linear systems assume a fixed value of $c$ (say, -1) \cite{10.1007/978-3-540-31954-2_19,10.1007/978-3-540-31954-2_17,10.1007/978-3-642-22110-1_30,10.1145/3049797.3049808}.
But in reality, the value of $c$ is not known exactly, instead, it is only estimated. Therefore, considering the sensing and actuation uncertainties, one should think of the value as a parameter that can take a range of values  \big(say, $c \in [-1.1,-0.9]$\big).
Often, the safety proofs provided using traditional reachable set based techniques cannot account for such uncertainties.
Further, the reachable set techniques do not inform us about the amount of uncertainties that can be tolerated while maintaining safety.


In this paper, we consider a model with uncertainties --- we define continuous time linear uncertain systems as $\dot{x}=\Lambda_u x$ (corresponding discrete time uncertain system is defined in \cite{10.1145/3358229}), where the uncertainties in the model are considered as variables in $\Lambda_u$ that can take values from a bounded range. 
%
%
Interval matrix, a matrix where elements are intervals, is a convenient representation of such a matrix $\Lambda_u$.
%
The linear uncertain system $\dot{x}=\Lambda_u x$ can equivalently be represented by separating the nominal and the perturbed dynamics as $\dot{x}=(A+\Lambda) x$, where $\Lambda_u=A+\Lambda$.

Our goal in this paper is to determine the uncertainties $\Lambda$ that can be tolerated by the matrix $A$ without violating the safety property, we call it the \emph{robustness threshold}.
For computing this, we propose a symbolic and a set-based (numeric) method to compute reachable sets that take into account the effect of uncertainties.
Second, we demonstrate the relationship between uncertainties and reachable sets using the maximum singular values.
Finally, we propose two heuristics to compute the robustness threshold.

With the aforementioned introduction to the problem of uncertainty in safety verification, we now, explicitly, layout the contributions of this paper.

\vspace{1ex}
\noindent \textbf{{\large Contributions:}}
\begin{enumerate}
    \item A symbolic approach, using classical approaches from perturbation theory, to compute reachable sets of linear uncertain systems. In previous works on uncertain linear systems, extensions of earlier proposed perturbation theory based methods have not been considered. To the best of the authors' knowledge, this is the first work to extend perturbation theory to compute reachable sets of linear uncertain systems.
    Given a continuous time linear system $\dot x=A x$, and an initial set $\theta$, the reachable set at time $t$ is given by $\texttt{RS}_{\theta}(A,t)$ = $e^{At} \theta_0$. But, if there is some perturbation $E \in \mathbb{R}^{n \times n}$ introduced to the system, $i.e.$, the system now becomes $\dot x=(A+E) x$, the reachable set $\texttt{RS}_{\theta}(A,t)$ becomes invalid, and the computation of $e^{(A+E)t}$ needs to be performed from scratch to get the reachable set of the perturbed system.
    In \cite{kaagstrom1977bounds} and \cite{10.2307/2156675}, closed form upper bounds on the sensitivity of matrix exponential to such perturbations are given. Given $A$, $E$ $\in \mathbb{R}^{n \times n}$, closed form upper bound $\phi_{A,E}(t)$ is computed. This $\emph{relative error bound}$ $\phi_{A,E}(t)$ can be used as a \emph{bloating factor} for $\texttt{RS}_{\theta}(A,t)$ to compute the reachable set of $\dot x=(A+E) x$, without the need to compute $e^{(A+E)t}$ from scratch.
    In this work, we define linear uncertain systems as $\dot x=(A+\Lambda)x$, where $\Lambda \subset \mathbb{R}^{n \times n}$ is an interval matrix. Therefore, to use the idea described above, we use the results from \cite{farhadsefat2011norms} on interval matrices and extend the results of \cite{kaagstrom1977bounds} and \cite{10.2307/2156675} on closed form upper bounds of $\phi_{A,E}(t)$ to interval matrices, $i.e.$, $\phi_{A,\Lambda}(t)$. This approach is discussed in Section \ref{sec:perturbation}.
    \item A set-based (numeric) approach to compute reachable sets of linear uncertain systems. For discrete time linear systems, the reachable set can be computed as a Minkowski sum of its autonomous part and the effect of the input (or perturbation) (\cite{girard2006efficient} and \cite{bak2017simulation}). If at a time step $k$, the reachable set of the autonomous part and the effect of input is $\Omega_k$ and $U_k$ respectively, then the cumulative reachable set of the system, at time step $k$, is given by $\Omega_k \oplus U_k$. We use this property of linear systems to compute the reachable set of a linear uncertain system. 
    Finally, we propose two methods to improve the performance of this approach. We discuss the details of this technique in Section \ref{sec:setBased}.
    %
    %
    %
    \item Establish a relationship between the reachable set of a linear uncertain system and the maximum singular value of the uncertain dynamics. 
    To the best of the authors' knowledge, this is the first work to establish such a relationship. 
    Leveraging this relationship, we propose two more useful artifacts that provide further insights on the behavior of linear uncertain systems: firstly, an algorithm to order the cells of the dynamics matrix based on its sensitivity to perturbation with respect to reachable sets; secondly, two heuristics to compute \textit{robustness threshold}.
    Informally, given a nominal system $\dot{x}=Ax$ and an unsafe set $\mathcal{U}$, the \textit{robustness threshold} is defined as the perturbation $\Lambda$, \textit{s.t}, the reachable set of the system $\dot{x}=(A+\Lambda)x$ does not intersect with the unsafe set $\mathcal{U}$, \textit{i.e.}, the system does not violate safety after the introduced perturbation. 
    \end{enumerate}
This paper is organized as follows.
In Section \ref{sec:relatedWork}, we provide related works. In Section \ref{sec:prelims}, we provide the notations and the formal definitions that has been used in the rest of the paper.
In Section \ref{sec:perturbation}, we discuss the perturbation theory based methods.
In Section \ref{sec:setBased}, we present a reachable set computation technique that uses the generalized-star representation, and present two techniques to improve its performance.
In Section \ref{sec:uncertainties}, we establish the relationship between the reachable set of linear uncertain systems and the maximum singular value of the uncertain dynamics matrix. We further provide an algorithm to order the cells of the dynamics matrix based on its sensitivity to perturbation, and two heuristics to compute robustness threshold.
In Section \ref{sec:caseStudies}, we provide two extensive case studies from varied domains --- a Mars Rover simulated on a real Martian Terrain, and an automated Anaesthesia delivery --- to evaluate all the artifacts proposed in this paper. We note that the purpose of this case study, from two varied domains, is not just limited to demonstrate the scalability of the proposed artifacts, but also demonstrate the applicability (and usability) of the artifacts on a variety of real life examples from robotics to medical devices. In Section \ref{sec:evaluation}, we further demonstrate our proposed artifacts on several linear dynamical system benchmarks.
We present our conclusions in Section \ref{sec:conclusion}.

\section{Related Work}
\label{sec:relatedWork}

In this section, we review the techniques for computing reachable sets of linear uncertain systems from the literature.
In \cite{10.1145/3358229}, a class of uncertainties are identified by a set of sufficient conditions on the structure of the dynamics matrix $\Lambda_u$. 
For such classes of uncertainties, the exact reachable set of the linear dynamical system can be computed very efficiently.
Unfortunately, this technique is unable to handle arbitrary set of uncertainties, whereas the techniques proposed in this paper can handle arbitrary set of uncertainties.
In \cite{DBLP:conf/emsoft/LalP15}, given a linear uncertain system $\dot{x}=\Lambda_u x$, where there are parametric uncertainties in $\Lambda_u$; the authors' provide a sampling interval $\delta>0$, given an $\epsilon>0$, $s.t.$, the piecewise bilinear function, approximating the solution by interpolating at these sample values, is within $\epsilon$ of the original trajectory. Contrary to \cite{DBLP:conf/emsoft/LalP15}, our techniques does not rely on sampling and therefore is independent of the size of uncertainties.
In \cite{althoff2011reachable}, given a continuous time linear system with input, the system is discretized and reachable set for consecutive time intervals are computed. 
For a given time step, the reachable set is computed in two steps steps: 
(i) For each matrix in the uncertain matrix, the \emph{state transition matrix} is expressed using the \emph{Peano-Baker} series. Then the approximate numerical solution to the series is obtained, iteratively, using \emph{Riemann sums}. 
(ii) Finally, using Zonotope representation, they propose a method to compute a convex-hull of the set that over-approximates the result of step (i) for all possible matrices in the uncertain matrix.
Unlike \cite{althoff2011reachable}, we propose two techniques, one using perturbation theory, and other using set-based representation where the reachable set is represented as Minkowski sum of its unperturbed part with the effect of perturbation.
In \cite{COMBASTEL20114525}, the authors provide an algorithm for computation of envelopes enclosing the possible states and/or outputs of a class of linear uncertain systems. 
It extends an existing algorithm based on Zonotopes so that it can efficiently propagate structured parametric uncertainties. 
Note that this work can handle only a class of uncertainties, but unlike \cite{10.1145/3358229}, they don't provide an exact reachable set --- the over-approximate flow-pipe has been empirically shown to handle the wrapping effect (\emph{i.e.} the tightness of over-approximation).

Unlike all these works mentioned, we propose: 
\begin{enumerate}
    \item A symbolic approach, using classical approaches from perturbation theory, which is computationally very fast.
    \item A set based representation (numeric) approach, where the reachable set is represented as a Minkowski sum of its unperturbed part with the effect of perturbation. 
    We further provide two methods to improve the performance of our set based representation method.
    \item We provide several other artifacts that provide more understanding on types and effects of various classes of uncertainties on a linear system.
    \item We provide two extensive case studies, from two different domains, to demonstrate the broad-spectrum applicability of the proposed artifacts. We further stress-test our methods on a battery of standard linear system benchmarks.
\end{enumerate}

Linear uncertain systems can also be modelled as a non-linear system. Some of the works that deal with computing reachable sets of non-linear systems are: \cite{7809839,10.1007/978-3-642-39799-8_18,6987596,10.1007/978-3-662-46681-0_5,10.1007/978-3-642-24690-6_13,10.1007/978-3-662-46681-0_15,10.1007/978-3-319-02444-8_37,ARCH15:An_Introduction_to_CORA}.

\section{Notations and Preliminaries}
\label{sec:prelims}

The state $x$ of the linear systems, considered in this work, is $x \in \mathbb{R}^n$.
Given $\delta > 0$, $B_{\delta}(x) = \{ y~|~ ||y-x|| \leq \delta \}$.
For a set $S \subset \mathbb{R}^n$, $B_{\delta}(S) = \cup_{x \in S} B_{\delta}(x)$. %
This operation over $S$ is called \emph{bloating} the set $S$ by $\delta$.
Given two sets $S_1$, $S_2$, its Minkowski sum is $S_1 \oplus S_2 = \{x_1 + x_2 ~|~ x_1 \in S_1, x_2 \in S_2\}$.
Given $i,j \in \mathbb{Z}$, with $i \leq j$, $[\![i, j ]\!]$ denotes the set $\{i, i+1, \ldots, j\}$.
A hyper-sphere with radius $k \in \mathbb{R}_{\ge 0}$ is denoted as $\mathcal{S}(k) \subset \mathbb{R}^{n}$; thus, a unit sphere is denoted as $\mathcal{S}(1)$.
%
%
%
We denote the usual notion of absolute value of $a \in \mathbb{R}$ as $|a| \in \mathbb{R}_{\ge 0}$.
We overload the operator $|.|$ for matrices as well --- given a matrix $M \in \mathbb{R}^{n \times m}$, $|M|$ is a matrix such that, $\forall_{i \in [\![1,n]\!],j \in [\![1,m]\!]}$ $|M|[i,j] = |M[i,j]|$. 
Given a matrix $M \in \mathbb{R}^{n\times n}$, the maximum singular value (Max SV) of $M$ is denoted as $\sigma(M)$, the determinant of $M$ is denoted as $det(M)$, and $e^{M} = I + \frac{M}{1!} + \frac{M^2}{2!} + \ldots$.

\begin{definition}
\label{def:cls}
\textbf{Continuous time Linear System}. Given a matrix $A \in \mathbb{R}^{n\times n}$, a continuous time linear dynamical system is denoted as $\dot{x} = Ax$. Trajectories of linear dynamical systems that give the state of the system at time $t$ starting from initial state $x_0$ is given as $\xi(x_0,t) = e^{At}x_0$.
\end{definition}


\begin{definition}
\label{def:intervalMatrix}
\textbf{Interval Matrix} (as defined in Section 2 of \cite{farhadsefat2011norms}). Given two matrices $A_{min}, A_{max} \in \mathbb{R}^{n \times n}$, \textit{s.t.}, $\forall i, j$ $A_{min}[i,j] \le A_{max}[i,j]$: we define an interval matrix $\Lambda \subset \mathbb{R}^{n \times n}$, also denoted as $[A_{min},A_{max}]$, where the $[i,j]^{th}$ element is the interval $[A_{min}[i,j], A_{max}[i,j]]$. The \emph{center} of $\Lambda$ (also denoted as $C_{\Lambda}$) is defined as $\frac{1}{2}(A_{min} + A_{max})$. Similarly, $\Delta_{\Lambda} = \frac{1}{2}(A_{max} - A_{min})$.
A matrix $E \in \Lambda$ if and only if $\forall i,j, E[i,j] \in \Lambda[i,j]$.
\end{definition}

Given two interval matrices $\Lambda_1,\Lambda_2 \subset \mathbb{R}^{n \times n}$, operations like addition, subtraction and multiplication are performed according to usual matrix algebra, where each operation is performed according to interval arithmetic. The notations are usual $\Lambda_1 \star \Lambda_2$, where $\star=\{+, \times, -\}$.
Given an interval matrix $\Lambda \subset \mathbb{R}^{n \times n}$, the singular value decomposition (SVD) of the interval matrix $\Lambda$ is given as $ \Lambda = \bigcup A = \bigcup U \Sigma V^*$, where $U,V$ are unitary matrices and $\Sigma$ is a diagonal matrix; we assume the decomposition is such that, the diagonal matrix $\Sigma$ is sorted in descending order.

\begin{definition}
\label{def:matNorms}
\textbf{Matrix Norms}. Given a matrix $M \in \mathbb{R}^{n \times n}$, its matrix 2 norm is denoted as $||M||_2$, and $||M||_F$ is the matrix Frobenius norm. We also use $||M||$ to denote $||M||_p$, where $p$ can be anything in $\{2, F\}$. For any matrix $M$, $||M||_2 \le ||M||_F$. 
\end{definition}

\begin{definition}
\label{def:matIntNorms}
\textbf{Interval Matrix Norm} (as defined in Theorem 1 of \cite{farhadsefat2011norms}). Given interval matrix $\Lambda$, $||\Lambda||_p=sup \{||A||_p~|~A \in \Lambda\}$, $p$ $\in$ $\{2,F\}$. We also use $||\Lambda||$ to denote $||\Lambda||_p$, where $p$ is obvious from the context.
\end{definition}

\begin{definition}
\label{def:clus}
\textbf{Linear Uncertain System}. An linear uncertain system is denoted as $\dot{x} = (A+\Lambda) x$, where $A \in \mathbb{R}^{n\times n}$ is a matrix and $\Lambda \subset \mathbb{R}^{n \times n}$ is an interval matrix. The set of states reached by an linear uncertain system after time $t$ is $\xi_{A+\Lambda}(x_0,t) = \{e^{(A+M)t}x_0~|~M \in \Lambda\}$. Here, $\Lambda$ denotes the model uncertainty. We also refer linear uncertain systems as uncertain linear system --- uncertain being the operative word. 
\end{definition}


One can extend the matrix exponential to linear uncertain systems as $e^{(A+\Lambda) t} = I + \frac{(A+\Lambda) t}{1!} + \frac{((A+\Lambda)t)^2}{2!} + \ldots$ where all the addition and multiplication operations are performed using interval arithmetic. It is trivial to observe that $\xi_{A+\Lambda}(x_0,t) = e^{(A+\Lambda) t} x_0$.


\begin{definition} \textbf{Max SV Candidate Matrix}.
\label{def:maxSVCand}
Given an interval matrix $\Lambda \subset \mathbb{R}^{n \times n}$, the Max SV Candidate $A_{\text{max}}$ is defined as $\Xi(\Lambda) = A_{\text{max}}$, where $\forall_{A \in \Lambda}~\sigma(A) \le \sigma(A_{\text{max}})$. Note that Max SV Candidate need not be unique. Let, $\sigma(A_{\text{max}})=\sigma_{\text{max}}$.
\end{definition}

\begin{definition} \textbf{Interval Matrix Transformation}.
Given a transformation $f(\theta) = \Lambda \theta$, where $\theta \subset \mathbb{R}^n$ and $\Lambda \subset \mathbb{R}^{n \times n}$, the transformation $f(\theta)$ is defined as 
$f(\theta)= \bigcup_{A \in \Lambda} A \theta$.
\end{definition}

\section{Reachability Using Perturbation Analysis}
\label{sec:perturbation}

In the early days of studying numerical techniques for solving ordinary differential equations, there was a significant interest in bounding the errors introduced by using finite precision representation of real numbers. Seminal works in this domain, such as \cite{kaagstrom1977bounds} and \cite{10.2307/2156675}, investigated analytical methods for computing upper bounds on the sensitivity of matrix exponential.
%
Given a linear dynamical system $\dot{x} = Ax$ and its perturbation $\dot{x} = (A+E)x$, these methods presented expressions to upper bound the relative distance between the trajectories of the nominal (unperturbed) and the perturbed system. The relative distance, denoted by $\phi_{(A,E)}(t)$, is given as:
\begin{equation}
\phi_{(A,E)}(t) = \frac{||e^{(A+E)t}-e^{At}||}{||e^{At}||}.
\end{equation}

Notice that \cite{kaagstrom1977bounds} and \cite{10.2307/2156675} were applicable for specific matrices $A$ and $E$, since we deal with bounded model uncertainties, we extend the same definition from a matrix $E$ to an interval matrix $\Lambda$. We extend the notation to include interval matrices as follows:
\begin{equation}
\phi_{(A,\Lambda)}(t) = sup_{E \in \Lambda} \frac{||e^{(A+E)t}-e^{At}||}{||e^{At}||}.
\end{equation}

Using the analytical expressions for $\phi_{(A,E)}(t)$ that were presented in \cite{kaagstrom1977bounds} and \cite{10.2307/2156675}, we provide analytical expressions for $\phi_{(A,\Lambda)}(t)$, involving interval matrices.
To compute $\phi_{(A,\Lambda)}(t)$, we leverage some of the results on interval matrix norms that are provided in~\cite{farhadsefat2011norms}.

\begin{lemma}
\label{lem:phi1}
\textbf{[From~\cite{kaagstrom1977bounds}, Table 4.1 (4.14)]}
Given a matrix $A$ and perturbation $E$ $\in$ $\mathbb{R}^{n \times n}$,
\begin{eqnarray}
&\phi_{(A,E)}(t) \leq  p_{n-1}(||A||_2t) \times \Big(exp\big(p_{n-1}(||A||_2t)||E||_2t\big)-1\Big). \label{eq:phi1}
\end{eqnarray}
where $p_{n-1}(x) = \Sigma_{k=0}^{n-1}\frac{x^k}{k!}$.
\end{lemma}

Notice from Equation~\ref{eq:phi1} that $\phi_{(A,E)}(t)$ is monotonically increasing with $||E||_2$. Therefore, to extend this analysis to interval matrices, one merely needs to compute the supremum of the set of 2-norms of all the matrices in the interval matrix.

\begin{lemma}
\label{lem:supremum}
$\sup\limits_{E \in \Lambda}~\{||E||_2\}=||\Lambda||_2$
\end{lemma}
As $||\Lambda||_2$ is obtained using Theorem 7 of~\cite{farhadsefat2011norms}, where $||\Lambda||_2$ is defined in Definition \ref{def:matIntNorms}, the above lemma holds true.

\begin{theorem}
\label{thm:supBloatForm}
$\phi_{(A,\Lambda)}(t) \leq p_{n-1}(||A||_{2}t)(exp(p_{n-1}(||A||_2 t)||\Lambda||_2 t)-1)$.
\end{theorem}
\begin{proof}
Observe that $\phi_{(A,E)}(t)$ (in Lemma \ref{lem:phi1}) monotonically increases with $||E||_2$, and Lemma \ref{lem:supremum} states that the 2 norm of the interval matrix $\Lambda$ is the supremum of 2-norms of all the matrices in the interval matrix. Therefore, we have $\phi_{(A,\Lambda)}(t) \leq p_{n-1}(||A||_{2}t)(exp(p_{n-1}(||A||_2 t)||\Lambda||_2 t)-1)$.
\end{proof}

We denote the upper bound provided in Theorem~\ref{thm:supBloatForm} as~\textbf{\texttt{Kagstrom1}}.

In the next two Theorems, we provide two different upper bounds of $\phi_{(A,\Lambda)}(t)$ than in Theorem \ref{thm:supBloatForm}.

\begin{lemma}
\label{lem:kag2}
\textbf{[From~\cite{kaagstrom1977bounds}, Table 4.1 (4.12)]} Given a matrix A and perturbation $E$ $\in$ $\mathbb{R}^{n\times n}$,
\begin{equation}
\phi_{(A,E)}(t) \leq K(SD)e^{\epsilon t}( e^{K(SD)\cdot||E||_2 t} - 1).
\end{equation}
where (a) for a matrix $M$, $K(M)$ is its condition number $|| M || \cdot ||M^{-1} ||$, (b) $S J S^{-1}$ is the Jordan form of the matrix $A$, and (c) $D$ is a diagonal matrix, \textit{s.t.}, $||D^{-1} J D ||_{2} \leq \epsilon$.
\end{lemma}

\begin{theorem}
\label{thm:supBloatForm2}
$\phi_{(A,\Lambda)}(t) \leq K(SD)e^{\epsilon t}( e^{K(SD)\cdot||\Lambda||_2 t} - 1)$;  where $K$, $S$, $D$, and $\epsilon$ are as defined in Lemma~\ref{lem:kag2}. 
\end{theorem}
\begin{proof}
Similar to the proof of Theorem~\ref{thm:supBloatForm} --- follows from monotonicity of $\phi_{(A,E)}(t)$ (in Lemma \ref{lem:kag2}) with respect to $||E||_2$, and Lemma \ref{lem:supremum}.
\end{proof}

We denote the upper bound provided in Theorem~\ref{thm:supBloatForm2} as \textbf{\texttt{Kagstrom2}}.

\begin{lemma}
\textbf{[From~\cite{10.2307/2156675}, Theorem 1]} Given a matrix $A$ and perturbation $E$ $\in$ $\mathbb{R}^{n \times n}$
\begin{equation}
\phi_{(A,E)}(t) \leq t ||E||_2 e^{(||A||_2 - \alpha(A) + ||E||_2)t}.
\end{equation}
Where, for any square matrix $M$, $\alpha(M)$ is the largest real component of the eigenvalues of $M$.
\end{lemma}

\begin{theorem}
\label{thm:supBloatLoan}
$\phi_{(A,\Lambda)}(t) \leq t ||\Lambda||_2 e^{(||A||_2 - \alpha(A) + ||\Lambda||_2)t}$
\end{theorem}
\begin{proof}
Similar to the proof of Theorem~\ref{thm:supBloatForm2}.
\end{proof}

We denote the upper bound provided in Theorem~\ref{thm:supBloatLoan} as \textbf{\texttt{Loan}}. 

We would like to inform the reader that \cite{kaagstrom1977bounds} and \cite{10.2307/2156675} contain many such formulas for computing an upper bound on $\phi_{(A,E)}(t)$. The authors focused on three of these specific formulas primarily because they yielded the best results on the standard verification benchmarks.
\emph{Theorems~\ref{thm:supBloatForm},~\ref{thm:supBloatForm2} and~\ref{thm:supBloatLoan} summarize the new contributions of our paper in extending the perturbation theory based techniques.} 
Computing the reachable set of a linear system with uncertainty is, therefore, a two step process.
First, we compute the reachable set of the dynamical system $\dot{x} = Ax$.
Second, we compute an upper bound on the relative error for the bounded uncertainties (i.e., $\phi_{(A,\Lambda)}(t)$) and \emph{bloat} the reachable set computed in the first step.
It is easy to prove that this procedure results in an over approximation of the reachable set for linear systems with uncertainties. The reachable set computation technique, proposed in this section, is computationally very efficient --- one merely needs to compute the bloating factors, as given in Theorems~\ref{thm:supBloatForm},~\ref{thm:supBloatForm2} and~\ref{thm:supBloatLoan}, and bloat the reachable set of the nominal dynamics to obtain the reachable set of the uncertain system. Our evaluation on several (low and high dimensional) benchmarks show that the reachable set of a linear uncertain system can be computed in just few seconds (less than a second for most of the benchmarks).

In the following sub-section, we extend our results (Theorems ~\ref{thm:supBloatForm},~\ref{thm:supBloatForm2} and~\ref{thm:supBloatLoan}) to use Interval Matrix Frobenius Norm, as it is computationally more efficient, compensating the quality of over-approximation.

\subsection{2-Norm and Frobenius Norm}
The 2-norm of an interval matrix $\Lambda$ can be computed using Theorem 7 of~\cite{farhadsefat2011norms} as follows:
\begin{equation}
||\Lambda||_2 = sup_{|y| = e_n, |z| = e_n} || C_{\Lambda} + (y z^{T}) \cdot \Delta_{\Lambda}||.
\end{equation}
where $e_n = (1, 1, \cdots, 1)^T$. Additionally, the elements of $y$ and $z$ can only be $\pm 1$. Therefore, computing $||\Lambda||_2$, in the worst case, requires computing 2-norm of $2^{2n}$ matrices. That is, $O(2^{2n} f_2(n))$, where $f_{2}(n)$ is the time complexity for computing 2-norm of a matrix. On the other hand, Frobenius norm, which upper bounds 2-norm, is easy to compute. From Theorem 10 of~\cite{farhadsefat2011norms}, we have:
$
||\Lambda ||_F = ||~ |C_{\Lambda}| + \Delta_{\Lambda} ||_{F}
$.
Hence, Frobenius norm of interval matrices can be computed by computing Frobenius norm of a single matrix. Therefore, using this, we can provide conservative upper bounds of \textbf{\texttt{Kagstrom1}}, \textbf{\texttt{Kagstrom2}} and \textbf{\texttt{Loan}} (Theorems \ref{thm:supBloatForm}, \ref{thm:supBloatForm2} and \ref{thm:supBloatLoan}) using Frobenius norm as follows: 

\begin{theorem}
\label{thm:supBloatFormFrob}
(Re-writing Theorem \ref{thm:supBloatForm})\\
$\phi_{(A,\Lambda)}(t) \leq p_{n-1}(||A||_{2}t)(exp(p_{n-1}(||A||_2 t)||\Lambda||_F t)-1)$.
\end{theorem}

\begin{theorem}
\label{thm:supBloatForm2Frob}
(Re-writing Theorem \ref{thm:supBloatForm2})
\\
$\phi_{(A,\Lambda)}(t) \leq K(SD)e^{\epsilon t}( e^{K(SD)\cdot||\Lambda||_F t} - 1)$ where $K$, $S$, $D$, and $\epsilon$ are as defined in Lemma~\ref{lem:kag2}. 
\end{theorem}

\begin{theorem}
\label{thm:supBloatLoanFrob}
(Re-writing Theorem \ref{thm:supBloatLoan})
\\
$\phi_{(A,\Lambda)}(t) \leq t ||\Lambda||_F e^{(||A||_2 - \alpha(A) + ||\Lambda||_F)t}$
\end{theorem}

In all our experiments, we compute 2-norm of interval matrices (Theorems \ref{thm:supBloatForm}, \ref{thm:supBloatForm2} and \ref{thm:supBloatLoan}) when it is easy to compute (for low dimensional systems). However, for high dimensional systems, we use Frobenius norm of interval matrices (Theorems \ref{thm:supBloatFormFrob}, \ref{thm:supBloatForm2Frob}, \ref{thm:supBloatLoanFrob}).

\section{Reachability Using Set Representations}
\label{sec:setBased}

In reachable set computation techniques, a suitable representation (such as Zonotopes~\cite{girard2005reachability}, Support functions~\cite{frehse2011spaceex}, or Stars~\cite{bak2017hylaa}) is chosen to represent the set of states reached by a dynamical system.
For efficient computation of reachable set, a time step $h$ is chosen, and the linear dynamical system $\dot{x} = Ax$ is converted into a discrete time dynamical system $x^{+} = e^{Ah}x$.
Given an initial set $\Theta$ in the representation, the linear transformation of $\Theta$, \textit{i.e.}, $e^{Ah} \Theta$ is computed in the same representation.

In the case of linear uncertain systems, \textit{i.e.}, $\dot{x} = (A+\Lambda)x$, we perform the same thing and compute a linear uncertain transformation as $x^{+} = e^{(A+\Lambda)h}x$.
Using widely known techniques for computing matrix exponential~\cite{moler2003nineteen,zutshi2012timed}, we compute a discrete time linear uncertain system, where $x^{+} = (\bar{A} + \bar{\Lambda})x$. Since the reachable set computation technique, proposed in this section, only relies on discrete transformation; in this section, we abuse notation and represent the discrete time linear uncertain system as $x^{+} = (A+\Lambda)x$.

In this paper, we use star representation~\cite{duggirala2016parsimonious} for computing reachable sets of linear uncertain systems.
We over-approximate the effect of model uncertainties as time varying environmental inputs, and represent the reachable set 
as a Minkowski sum of the reachable set without inputs and the effect of inputs, as presented in~\cite{bak2017simulation}.
We also provide two new methods to improve the scalability of this approach and evaluate this method on a wide variety of benchmarks.
%

\begin{definition}
\label{def:stars}
A star is a tuple $\langle a, G, P \rangle$, where $a \in \mathbb{R}^n$ is called \emph{anchor}, $G$ is a set of vectors $\{g_1, g_2, \ldots, g_m\}$ where $g_i \in \mathbb{R}^n$ is called \emph{generators}, and $P:\mathbb{R}^{m} \rightarrow \{\top, \bot\}$ is called \emph{predicate}; represents a set of states. Formally, 
$
\langle a, G, P \rangle 
\deq
\{ x~|~ \exists \alpha_1, \ldots, \alpha_m, x = a + \Sigma_{i=1}^{m} \alpha_i g_i, P(\alpha_1, \ldots, \alpha_m) = \top \}.
$
\end{definition}

Given a matrix (linear transformation) $A$, it is easy to see that $A\cdot \langle a,G,P \rangle = \langle A\cdot a, A\cdot G, P \rangle$, where $A\cdot G = \{A\cdot g_1, \ldots, A\cdot g_m\}$.
That is, linear transformation of a star set can be computed extremely quickly without changing the predicate of the star.
Additionally, given two stars $\Theta_1 = \langle a_1, G_1, P_1 \rangle$ and $\Theta_2 = \langle a_2, G_2, P_2 \rangle$, the Minkowski sum in star representation $\Theta_1 \oplus \Theta_2 = \langle a_1 + a_2, G_1 \cup G_2, P_1 \land P_2 \rangle$. %
In the Minkowski sum, we ensure that the variables in $P_1$ and $P_2$ are different, and the $\alpha$-s corresponding to $P_1$ ($P_2$) are coefficients for vectors in $G_1$ ($G_2$ respectively).

\subsection{Computation of Reachable Sets for Linear Uncertain Systems using Stars}
\label{sec:uncertainstars}

To perform the linear uncertain transformation $x^{+} = (A + \Lambda)x$ over a star $\langle a, G, P \rangle$, our approach considers $\Lambda$ as an additional input.
That is, $(A + \Lambda)\cdot \langle a, G, P \rangle$ is computed as $A \cdot \langle a, G, P \rangle \oplus \Lambda \cdot \langle a, G, P \rangle$.
The former part of the summation \big($A \cdot \langle a, G, P \rangle$\big) is easier to compute.
For the latter part, $\Lambda \cdot \langle a, G, P \rangle $, we impose a specific assumption on $P$, and compute a set $u$, \textit{s.t}, $\Lambda \cdot \langle a, G, P \rangle \subseteq u$. 
Note that, as $\Lambda$ is an interval matrix, $\Lambda \cdot \langle a, G, P \rangle$ is a set of generalized stars, $i.e.$, $\Lambda \cdot \langle a, G, P \rangle = \bigcup_{E \in \Lambda} \langle E \cdot a, E \cdot G, P \rangle$. Using Theorem \ref{thm:compU}, we over-approximate $\Lambda \cdot \langle a, G, P \rangle$ (a set of generalized star) with a generalized star $u$, $i.e.$, $\Lambda \cdot \langle a, G, P \rangle \subseteq u$.

\begin{theorem}
\label{thm:compU}
Given an interval matrix $\Lambda$ and a star $\langle a, G, P \rangle$, let $\Lambda \cdot \langle a, G, P \rangle = \bigcup_{E \in \Lambda} E \cdot \langle a, G, P \rangle$. Also consider that $\forall \alpha_1, \ldots, \alpha_m$ such that $P(\alpha_1, \ldots, \alpha_m) = \top$, the upper and lower bounds on $\alpha_j$ are $c_{low}^{j}$ and $c_{high}^{j}$. We have $\Lambda \cdot \langle a, G, P \rangle \subseteq u$, where $u = \langle \bar{0}, I_{n}, P' \rangle$: (a) $\bar{0}$ represents the origin, (b) $I_n$ is the collection of $n$ orthonormal unit vectors that are axis aligned, and (c) $P'$ is a conjunction of box aligned constraints such that $d_{low}^{i} \leq \alpha'_i \leq d_{high}^{i}$, where $[d_{low}^{i}, d_{high}^{i}]$ is the interval obtained by performing $\Lambda \cdot a [i] + \Sigma_{j=1}^{m} [c_{low}^j , c_{high}^j] \cdot \Lambda \cdot g_j$ using interval arithmetic.
\end{theorem}
\begin{proof}
Using the expansion of the definition of generalized star (Definition~\ref{def:stars}) and the provable over-approximation of interval arithmetic, we use contradiction to prove this theorem. Let there be an $ 1 \le i \le n$ for which $\Lambda \cdot \langle  a,  G, P\rangle$ $\not\subseteq$ $\langle \bar{0}, I_{n}, P' \rangle$. This would imply $[d^i_{low}, d^i_{high}]$ $\not\subseteq$ $[low(\alpha'_i),high(\alpha'_i)]$, which violates the assumption.
\end{proof}

Intuitively, Theorem~\ref{thm:compU} over-approximates $\Lambda \cdot \langle  a,  G, P\rangle$ with an axis-aligned hyperbox $u$, $i.e.$, $\Lambda \cdot \langle a, G, P \rangle \subseteq u$. Therefore, using Theorem \ref{thm:compU}, given an interval matrix $\Lambda$ and a star $\langle  a,  G, P\rangle$, one can easily compute an $u$, \textit{s.t}, $\Lambda \cdot \langle a, G, P \rangle \subseteq u$.
Using Theorem~\ref{thm:compU}, we can over-approximate the linear uncertain transformation (\textit{i.e.}, $A+\Lambda$) of a star $\theta$ as a Minkowski sum of two stars.
Let $A\cdot \theta \oplus u_0$ be the over-approximation of the reachable set after one transformation, where $\Lambda \cdot \theta \subseteq u_0$ ($u_0$ is obtained using Theorem \ref{thm:compU}).
Applying the transformation of $(A + \Lambda)$ for the second time would yield $A^2 \cdot \theta \oplus A \cdot u_0 \oplus u_1$, where $u_1$ is obtained by applying Theorem~\ref{thm:compU} on $\Lambda (A \cdot \Theta \oplus u_0)$.
The recursive relationship that represents the over-approximation of the reachable set, at time step $k$, is given as follows:
\begin{eqnarray}
ORS_{k} &=& \theta \mbox{ (initial set) }~~~~~~~~~\mbox{if $k$ = 0} \nonumber \\
&=& A\cdot ORS_{k-1} \oplus u_k ~~~~\mbox{ otherwise.} \label{eq:ORScomp}
\end{eqnarray}
This is graphically represented in Figure~\ref{fig:methodORS}. It also follows from Theorem~\ref{thm:compU} that $(A+\Lambda)^k\cdot \theta \subseteq ORS_{k}$.

\begin{figure}
    \includegraphics[height=0.25\textheight]{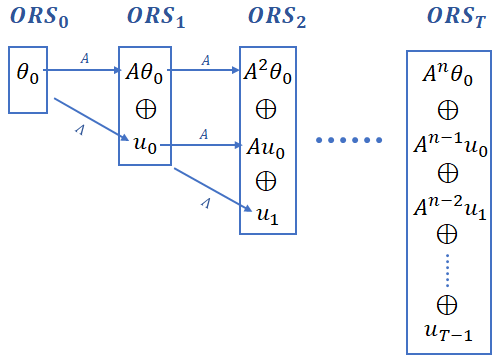}
    \caption{Recurrence relation to compute an over-approximate reachable set of a linear uncertain system, using Set Representation method.}
    \label{fig:methodORS}
\end{figure}

\subsection{Methods to improve performance}
\label{subsec:performanceImprove}
The set representation based method performs Minkowski sum of two stars, at every step, to compute the over-approximate reachable set.
The reachable set at time step $k+1$ is given by $ORS_{k+1} = A \cdot ORS_k \oplus u_k$, due to the Minkowski sum of two generalized stars, at every time step, we increase the number of basis vectors for representing the over-approximate reachable set by $n$ ( where $n$ is the dimension of the system). The number of basis vectors increases by $n$ at each step, because $u_k$ is represented using $n$ basis vectors (as it is computed using Theorem \ref{thm:compU}). 

Due to the increase in the number of basis vectors of the generalized star representing the over-approximate reachable set at each step, the computation of $u_k$, using Theorem \ref{thm:compU}, becomes computationally expensive. Say, the generalized star representing the over-approximate reachable set $ORS_k$, at time step $k$, has $m$ basis vector, where $m>>n$: then the computation of $u_k$ (using Theorem \ref{thm:compU}), \textit{s.t}, $\Lambda \cdot ORS_k \subseteq u_k$, will involve computing maximum and minimum value of $\Lambda \cdot a [i] + \Sigma_{j=1}^{m} [c_{low}^j , c_{high}^j] \cdot \Lambda \cdot g_j$, where $j \in [\![1,m]\!]$, $i \in [\![1,n]\!]$, and $\Lambda$ is an interval matrix. As the number of basis vectors $m$, representing the reachable set, increases at every step (increases by $n$ at each step), the number of interval arithmetic operations needed to compute $u_k$ also increases.
Hence, reducing the number of generators would significantly improve the performance of the reachable set computation.
We present the following two methods for this purpose.

\subsubsection{Interval Arithmetic Based Reduction}
In our first method, given $ORS_k$, we project the generators (along with the coefficients) of $ORS_k$ on the orthonormal axis and compute $n$ ($n$ is the dimension of the system) axis aligned generators. That is, we compute an axis aligned hyperbox (represented as a generalized star with $n$ basis vectors) that over-approximates $ORS_k$.
Given a generalized star $\mathcal{S}_m= \langle a, G, P\rangle$ (representing a reachable set of an $n$ dimensional system) with $m$ basis vectors ($m>>n$), we use Theorem \ref{thm:intervalRed} to compute a generalized star $\mathcal{S}_n$, with $n$ basis vectors, $s.t$, $\mathcal{S}_m \subseteq \mathcal{S}_n$.

\begin{theorem}
\label{thm:intervalRed}
Let a generalized star be $\mathcal{S}_m=\langle a, G, P\rangle$; where $G$ is a set of $m$ basis vectors, and let $P$ be of the form $c_{low}^j \le \alpha_j \le c_{high}^j$. Then $\mathcal{S}_m=\langle a, G, P\rangle$ $\subseteq$ $\mathcal{S}_n$; where $\mathcal{S}_n = \langle \bar{0}, I_{n}, P' \rangle$: (a) $\bar{0}$ represents the origin, (b) $I_n$ is the collection of $n$ orthonormal unit vectors that are axis aligned, and (c) $P'$ is a conjunction of box aligned constraints such that $d_{low}^{i} \leq \alpha'_i \leq d_{high}^{i}$, where $[d_{low}^{i}, d_{high}^{i}]$ is the interval obtained by performing $a [i] + \Sigma_{j=1}^{m} [c_{low}^j , c_{high}^j] \cdot g_j$ using interval arithmetic.
\end{theorem}
\begin{proof}
Similar to the proof of Theorem \ref{thm:compU}
\end{proof}
Intuitively, given a generalized star $\mathcal{S}_m$ with $m$ basis vectors, Theorem \ref{thm:intervalRed} computes an axis aligned hyperbox $S_n$, represented as a generalized star with $n$ basis vectors, $s.t$, $\mathcal{S}_m \subseteq \mathcal{S}_n$. This operation is used to reduce the number of generators from $m$ to $n$.
Since the Minkowski sum increases the number of generators by $n$ at each step, applying this technique at regular intervals would yield quadratic improvement in performance.
However, it also results in increased over-approximation of the reachable set.

\subsubsection{Zonotope Based Reduction}
Our second method for reducing the number of generators uses a tool for performing operations on polytopes, called \textbf{\texttt{pypolycontain}} \cite{sadraddini2019linear}\footnote{\url{https://github.com/sadraddini/pypolycontain}}.
In~\cite{sadraddini2019linear}, the authors propose a representation for polytopes that makes checking of polytope containment very efficient. 
As a side effect, it also provides an efficient Zonotope reduction method that uses linear programming.
When provided with an initial set of generators and a Zonotope as input, the tool performs refinement and returns a Zonotope that has same number of generators as the provided initial set of generators and over-approximates the input Zonotope.
Given a generalized star $\mathcal{S}_m=\langle a, G, P \rangle$ representing a reachable set of a $n$ dimensional system, where $G$ is a set of $m$ basis vectors ($m >> n$); using this open source tool, we obtain a Zonotope (therefore a generalized star) $\mathcal{Z}_n$, with $n$ basis vectors, $s.t.$ $\mathcal{S}_m \subseteq \mathcal{Z}_n$.
We modified this tool to take our generalized star as input and return a generalized star as output.
Similar to the first method, this method also gives a quadratic improvement in performance and increases the order of over-approximation of the reachable set.
Additionally, the over-approximation is also dependent on the initial set of generators provided to the tool.

\section{All Uncertainties Are Not Equal}
\label{sec:uncertainties}

\subsection{Relationship between Reachable Sets and Singular Values}
\label{subsec:2xProof}

In this section, we demonstrate the relationship between the reachable set of a linear system and its corresponding singular values.
We first show that the widening of the reachable set (along particular directions), when viewed as a linear transformation, is dependent on the singular values.
We then show that a template based overapproximation of the reachable set is dependent on the maximum singular values.
This relationship between the widening of the reachable set and the singular values would help us understand the effect of specific uncertainties on the reachable set.
%


\begin{figure}
    \includegraphics[ height=0.25\textheight]{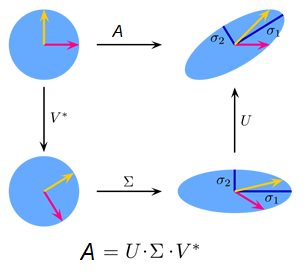}
    \caption{Geometric interpretation of SVD. \cite{enwiki:1013751233}}
    \label{fig:svd}
\end{figure}


\begin{remark}
Given a set $\theta=\mathcal{S}(1)$ (hyper-sphere of unit radius) and a linear transformation $A$, the change in the width of the set $\theta$ (in particular directions) after the linear transformation, is equal to the singular values of $A$.
That is, the width of the set $A \theta$, in evaluated directions, is equal to the singular values.
\end{remark}

Consider the singular value decomposition $U \Sigma V^*$ of the matrix $A$. 
The linear transformation can be understood as a sequence of three linear transformations $V^*$, followed by $\Sigma$, followed by $U$. 
Since $U$ and $V^*$ are unitary matrices, the transformation does not cause the set to widen in any direction. 
Whereas, the widening of the set $A\theta$, in evaluated directions from $U$ and $V^*$, to the the set $\theta$ is equal to the singular values of $A$. 
More specifically, the transformation $A \theta$ is an ordered sequence of rotation, scaling, and rotation. That is, $V^*$ is a rotation, $\Sigma$ is a scaling, and $U$ is a rotation.
This is illustrated in Figure~\ref{fig:svd}.

We now show that a template based overapproximation of the set $A \theta$ is proportional to the maximum singular value of $A$. 
We illustrate this by performing a linear transformation on a unit sphere $\sphere(1)$ and show that the overapproximation of $A \sphere(1)$ is contained within a sphere of radius $\sigma^A_{max}$, the maximum singular value of $A$.

\begin{lemma}
\label{lem:maxvol}
Given the unit sphere $\sphere(1)$ and a matrix $A$, we have $A \times \sphere(1) \subseteq \sphere(\sigma^A_{max})$, where $\sigma^A_{max}$ is the maximum singular value of the matrix $A$.
\end{lemma}
\begin{proof}
This proof hinges on the geometric interpretation of the singular value decomposition. 
Consider the SVD of $A$ as $U \Sigma V^*$. Since a unit sphere is symmetric, the linear transformation of $V^* \sphere(1)$ ($i.e.$ a rotation) is same as $\sphere(1)$.
The transformation of $\Sigma$, a diagonal matrix with singular values, scales (widens) the $i^{th}$ axis by the corresponding singular value $\sigma_i$, resulting in an ellipsoid whose axes are defined by the singular values.
The transformation of the ellipsoid with $U$ ($i.e.$ another rotation) will reorient the axes but not cause any  further widening.
Therefore, the maximum distance of any point on $A \sphere(1)$ is bounded by $\sigma^A_{max}$, the maximum singular value of $A$.
Hence, $A \sphere(1) \subseteq \sphere(\sigma^A_{max})$. This has been illustrated in Figure \ref{fig:sigmaBall}.
\end{proof}

In this paper we are interested in understanding the effect of matrix uncertainties on the reachable set. 
Given an uncertain system $\dot{x} = \Lambda x$, from Lemma~\ref{lem:maxvol}, we know that the matrix with the maximum singular value of the matrix exponential would have the largest effect vis-à-vis widening on the reachable set --- this has been formalized in Theorem \ref{thm:2x}.
\emph{Therefore, the uncertainties that will increase the maximum singular value would cause an increase in the width of the overapproximation of the reachable set.}
Therefore, for evaluating the robustness of safety of a linear dynamical system, we introduce perturbations that cause largest change in the maximum singular value. In the following theorem, we show that the reachable set, viewed as a transformation, can be computed using just the maximum singular value of the Max SV Candidate (Definition \ref{def:maxSVCand}).

\begin{theorem}
\label{thm:2x}
Given a transformation $f(\theta)=\Lambda \theta$ where $\Lambda \subset \mathbb{R}^{n \times n}$; and $\theta = \sphere(1)$, then $f(\theta) \subseteq \sphere(\sigma_{\text{max}})$.
\end{theorem}
\begin{proof}
Follows from Lemma~\ref{lem:maxvol}.
\end{proof}

\subsection{Ordering the cells of the matrix based on its sensitivity to perturbation}
\label{subsec:ordering}
In this section, we study the effect of introducing uncertainties in different cells of a linear dynamical system. 
%
From Theorem \ref{thm:2x} we know that by studying the effect of an uncertainty on the singular value, we can understand the effect of that uncertainty on the magnitude of change in the reachable set.

The effect of perturbations on the singular values of a matrix has, fortunately, already been studied in~\cite{it:1999-001}.
In~\cite{it:1999-001}, given matrices $A$ and $B$, the singular values of the matrices $A + \epsilon B + O(\epsilon^2)$ are shown to have the form $\sigma_i + k_i \epsilon + O(\epsilon^2)$. 
Additionally,~\cite{it:1999-001} also provides the closed form expression to compute $k_i$ as a function of $A$ and $B$.
Therefore, the matrix with largest $k_i$ causes the largest effect in the singular values as a result of perturbation.

\begin{figure}
    \includegraphics[height=0.32\textheight]{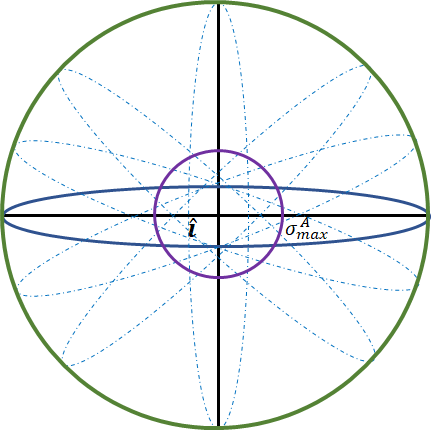}
    \caption{Illustration of the overapproximation being contained in $\sphere(\sigma^A_{max})$. Purple hypersphere: $\sphere(1)$ the initial set that is invariant to transformation by unitary matrix. Blue ellipse: $\sphere(1)$ after applying the transformation $\Sigma$ causes scaling along various dimensions. Green hypersphere: $\sphere(\sigma^A_{max})$ that contains all possible transformations $U$ of the blue ellipse.}
   \label{fig:sigmaBall}
\end{figure}

Using the closed form expression provided in~\cite{it:1999-001}, we can introduce a perturbation in each cell of $A$ and study its effect on the singular value.
%
To introduce a perturbation of $\epsilon$ in the $[i,j]^{th}$ element of the matrix $A$ and observe the change in singular values, we use the expression in~\cite{it:1999-001} with matrix $B$ as a zero matrix except for the $[i,j]^{th}$ element being $1$.
A similar technique for constructing $B$ where $B[i,j] = A[i,j]$ can be used to study the effect of relative perturbation on singular values.
%
%
If perturbation in a cell $(i,j)$ has more impact on the singular value \textit{vis-\`a-vis} the reachable set than a cell $(k,l)$, then Algorithm~\ref{algo:order} orders $(i,j)$ before $(k,l)$.

Algorithm~\ref{algo:order} ranks the cells of the matrix in decreasing order of relative change in the singular values by perturbation. 
That is, a perturbation in the first element of $\mathit{SingularOrder}$ causes the maximum change in the singular values of the matrix $A$ (thus on the reachable set). 
Since we are interested in \emph{relative perturbation}, \textit{i.e.}, the elements of the matrix $A$ are subjected to relative uncertainty ( such as $1\%, 5\%$, etc.), we initiate $B[i,j] = A[i,j]$ (and rest of the elements as 0) in Algorithm~\ref{algo:order}. 

\begin{algorithm}
\SetKwInOut{Input}{input}
\SetKwInOut{Output}{output}
\Input{Linear Dynamics $A$.}
\Output{An ordering on the cells $[i,j]$ of the matrix according to the effect on perturbation on maximum singular value of $A$.}
$B \leftarrow \textit{zeros}[n][n]$\;

\For{ $ 1 \leq i,j \leq n$}{
  $B[i,j]=A[i,j]$\;
  $Ord[i,j] \gets \mathsf{maximumSingularValueChange}(A,B)$ \CommentSty{\big(Eq (3.8) in \cite{it:1999-001}\big)} \; \label{ln:maxpert}
  $B[i,j]=0$\;
}

$\mathit{SingularOrder} \leftarrow$ Sort $(i,j)$ based on decreasing value of $Ord[i,j]$; \\
\Return{$\mathit{SingularOrder}$};
\caption{Algorithm to find an ordering on the cells of the matrix $A$ with decreasing sensitivity to perturbation}
\label{algo:order}
\end{algorithm}

%
For each cell $[i,j]$ of the matrix, Algorithm~\ref{algo:order} constructs the matrix $B$ that is all zeros, except for the cell $[i,j]$ and invokes the closed form expression provided in~\cite{it:1999-001} with the pair $(A,B)$.
This expression will determine the change in maximum singular value of $A$ after applying $\epsilon$ relative perturbation.
The magnitude of this change is stored in $Ord[i,j]$ (line~\ref{ln:maxpert}).
All the cells of the matrix are sorted according to the decreasing order of the effect on maximum singular value.



\subsection{Computing Robustness Threshold}
\label{subsec:robusntess}

We define \textit{robustness threshold} of a linear dynamical system as: \emph{given a set of cells of the dynamics matrix and an unsafe set, the maximum amount of perturbation that can be introduced to the dynamics while being safe}. 

\begin{definition}
Given a dynamical system $\dot{x} = Ax$, initial set $\Theta$, unsafe set $U$, and time bound $T$, the robustness threshold $r$ is given as:
\begin{equation}
r = max_{||\Lambda||} \{ \Lambda \subseteq \mathbb{R}^{n \times n } \:|\: \forall A', A' \in A+\Lambda, \forall t, t \in [0, T], e^{A't}\Theta \cap U = \emptyset \} 
\label{eq:robustness}
\end{equation}
\end{definition}

Given that matrix exponential is a nonlinear function, performing the quantifier elimination and computing the exact value of robustness threshold \emph{exactly} is impractical.
In this section, we provide empirical techniques to approximate the value of robustness threshold.
Our procedure follows a numerical search, i.e., we first determine an uncertainty $\Lambda$ such that the uncertain linear system is safe and then progressively increase the uncertainties, until a safety violation is discovered.
Instead of computing the matrix exponential for all the matrices in the uncertain linear system, we compute the overapproximation of the reachable set using symbolic (Section \ref{sec:perturbation}) and numerical (Section \ref{sec:setBased}) methods.

We provide two heuristics for numerically estimating the value of robustness threshold in this section.
%
Both the proposed heuristics rely on the effect of a perturbation on the the maximum singular value as computed in Algorithm \ref{algo:order}. 
We first pick an initial value for the budget of perturbation (say 1\% relative perturbation to be distributed in the given cells).
Given a bound on the total perturbation that can be introduced in a dynamics, the perturbation should be distributed such that: the cells that have larger effect on the maximum singular value receive a smaller proportion of the total  perturbation and vice-versa. 
Recall that $Ord$ matrix, computed in Algorithm \ref{algo:order}, measures the effect of the cell $[i,j]$ on the maximum singular value.
Therefore, higher the value of $Ord[i,j]$, lower is the amount of perturbation that can be tolerated while keeping the system safe.
%
%
%

We distribute this budget among the various cells according to two heuristics.
%
%
%
The first heuristic, \tnm{Proportional Distribution} is such that the perturbation introduced in cell $[i,j]$ is proportional to the values of $Ord[i',j']$ that are less than $Ord[i,j]$. 
%
The second heuristic, \tnm{Harmonic Distribution} is such that the perturbation introduced in cell $[i,j]$ is proportional to $\frac{1}{Ord[i][j]}$. 
%
We use the baseline of \tnm{Equal Distribution} where the same amount of perturbation is introduced in all cells, irrespective of their effect on maximum singular value.

Once the perturbations for each cell is allocated, we compute the overapproximation of the reachable set for the uncertain linear system and infer safety. 
If the uncertain system is safe, we increase the budget by a constant value (say addition of $10\%$) for perturbation and continue.
%
%
We terminate when any further increase in the perturbation budget will cause the system to be unsafe.
This procedure is formally described in Algorithm~\ref{heu:robustMet}.

\begin{algorithm}
\SetKwInOut{Input}{input}
\SetKwInOut{Output}{output}
\Input{Linear Dynamics $A$, Set of cells $\mathcal{C}$ of $A$, Unsafe set $\mathcal{U}$, time step $t$}
\Output{Compute the robustness metric}

$p \gets$ 0  \CommentSty{(perturbation budget)} \;
$step \gets$ \text{user chosen parameter} \CommentSty{(increase in perturbation after each iteration)} \;
$\Lambda_{old}$=$\{\}$\;

\While{True}{
  $\Lambda \gets$ Perturb cells $A$ according to budget $p$ and user determined distribution\;
  \If{$\mathcal{U} \cap ORS_t(A+\Lambda) \neq \emptyset$}
  {return $||\Lambda_{old}||$\;}
  $p=p+step$\;
  $\Lambda_{old}=\Lambda$
}
\caption{Heuristic to compute robustness metric of linear dynamics}
\label{heu:robustMet}
\end{algorithm}

\section{Case Studies}
\label{sec:caseStudies}
To demonstrate the applicability (and usability) of the artifacts, proposed in this paper, on a variety of real life examples from robotics to medical devices; we implemented our algorithms and heuristics in a python-based prototype tool.
We will open source our tool if the paper gets accepted.
Our tool uses \texttt{numpy}~\cite{oliphant2006guide}, \texttt{scipy}~\cite{2020SciPy-NMeth}, \texttt{mpmath} \cite{mpmath} for matrix multiplications, Gurobi~\cite{gurobi} engine for visualization of the reachable sets, and \texttt{PythonRobotics} \cite{pythonrobotics} for various robotics modules such as Spline Planner and Model Predictive Control (MPC).
All of our experiments were performed on a Lenovo ThinkPad Mobile Workstation with
i7-8750H CPU with 2.20 GHz and 32GiB memory on Ubuntu 18.04 operating system (64 bit).

In this section, we evaluate the various artifacts, proposed in this paper, on the two following case studies from two different domains.

\subsection{Anaesthesia Delivery --- A PK/PD Model}
\label{subsec:pkpd}
Over or under use of anaesthesia can be be fatal to the patient --- over-dose can cause long term detrimental effect, and under-dose might fail to maintain the state of hypnosis, causing a traumatic experience to the patient \cite{ARCH15:Benchmark_Problem_PK_PD_Model}. Though anaesthesia has traditionally been performed manually, with improvements in healthcare infrastructure, there is a push for automating the process; thereby raising a need to verify such safety critical systems. With the aforementioned motivation, in this case study, we study a commonly used anaesthetic drug, propofol, \cite{ARCH15:Benchmark_Problem_PK_PD_Model} vis-à-vis safety with perturbation in model parameters.

\vspace{1ex}
\noindent \textbf{{\large Model:}}
The model considered in this case study has two components, PK and PD: the pharmacokinetics (PK) models the change in concentration of the drug as the body metabolizes it, and pharmacodynamics (PD) models the effect of the drug on the body. The model has three compartments, and the state variables track the concentration of each compartment: (i) the first peripheral compartment $c_1$, (ii) the second peripheral compartment $c_2$, (iii) the plasma compartment $c_p$. The input to the model is the infusion rate of the drug (propofol). The evolution of the states are dependent on several parameters, such as: the weight of the patient ($weight$), the first order rate constants between the compartments, $k_{10}$, $k_{12}$, $k_{13}$, $k_{21}$ and $k_{31}$ \cite{10.1093/bja/aei567}. The aforementioned compartments --- $c_1$, $c_2$ and $c_p$ --- comprise the pharmacokinetics (PK). The pharmacodynamics (PD), on the other hand, is tracked by the state variable $c_e$, which evolves relying on the parameter $k_d$ (the rate constant between
plasma and effect site) and other state variables. The detailed model can be found in Equation 5 of \cite{ARCH15:Benchmark_Problem_PK_PD_Model}. The system is considered safe --- that is, no under/over dose occurs --- if the following concentration levels are maintained: $c_1 \in [1,10]$, $c_2 \in [1,10]$, $c_p \in [1,6]$ and $c_e \in [1,8]$, where the input $u \in [0,200]$.

In this case study, we study the following questions: 
\begin{enumerate}
    \item Impact on the concentration levels of the first and the second peripheral compartment, $c_1$ and $c_2$ respectively, due to small perturbations in parameters $k_{21}$ and $k_{31}$. 
    \item Impact on the concentration level of the plasma compartment, $c_p$, due to small errors in the weight measurement of the patient. 
    \item Impact on the concentration level of the PD component, $c_e$, due to small perturbations in the parameter $k_d$. We answer questions (1) to (3) using the artifacts developed in Section \ref{sec:setBased}. 
    \item Which parameter has more impact on the concentration level of the first peripheral compartment ($c_1$), $k_{21}$ or $k_{31}$? We answer this question using Algorithm \ref{algo:order}. 
    \item The amount of error in the weight measurement of the patient that can be tolerated without violating the safe concentration level of the plasma compartment $c_p$; we answer this question using Heuristic \ref{heu:robustMet}.
    \end{enumerate}

The answers to the above questions, (1) to (5), establish connections between the concentration levels of the various compartments and the model parameters with regards to perturbation; thus helping the practitioners/designers to be mindful of any unfortunate parameter or measurement errors that might occur during the process.

\vspace{1ex}
\noindent \textbf{{\large Results:}}
The dynamics given in Equation 5 of \cite{ARCH15:Benchmark_Problem_PK_PD_Model} was discretized (with a step size 0.01) and the answers to questions (1)-(5) were obtained.
\begin{itemize}
    \item \textit{Answer to Question (1)}. To answer this question, we introduced a small perturbation of $\pm$1.8\% in the parameters $k_{21}$ and $k_{31}$. We observed that starting with an initial concentration level of $c_p \in [2,4]$, $c_1 \in [4,6]$, $c_2 \in [4,6]$, $c_e \in [3,5]$ and $u \in [0,10]$, the system reaches an unsafe concentration level in the first peripheral compartment $c_1$ in 20 time steps, whereas the unperturbed system ($i.e.$, with no perturbation) remains safe. The reachable sets, up-to 20 time steps, have been shown in Fig. \ref{fig:pkpd_c1c2}: The $x$ and the $y$ axis represents the concentration levels $c_1$ and $c_2$ respectively. Green represents reachable sets of the unperturbed system, and blue represents the reachable set of the perturbed system.
    This analysis, using the reachable set computation method proposed in Section \ref{sec:setBased}, took 0.32 seconds.
    \item \textit{Answer to Question (2)}. We introduced a small perturbation of $\pm$0.8\% in the weight measurement of the patient, and observed that starting with an initial concentration level of $c_p \in [2,4]$, $c_1 \in [3,6]$, $c_2 \in [3,6]$, $c_e \in [2,4]$ and $u \in [2,10]$, the system reaches an unsafe concentration level in the plasma compartment $c_p$ in 20 time steps, whereas the unperturbed system remains safe. The reachable sets, up-to 20 time steps, have been shown in Fig. \ref{fig:pkpd_cp}: The $x$ and the $y$ axis represents the time steps and the concentration level $c_p$ respectively. Cyan represents reachable sets of the unperturbed system, and blue represents the reachable set of the perturbed system.
    This analysis, using the reachable set computation method proposed in Section \ref{sec:setBased}, took 0.32 seconds.
    \item \textit{Answer to Question (3)}. We introduced a perturbation of $\pm$5\% in the parameter $k_d$, and observed that starting with the same initial concentration level as before, the system reaches an unsafe concentration level in the PK component $c_e$ in 20 time steps, whereas the unperturbed system remains safe. The reachable sets, up-to 20 time steps, have been shown in Fig. \ref{fig:pkpd_ce}: The $x$ and the $y$ axis represents the time steps and the concentration level $c_p$ respectively. Cyan represents reachable sets of the unperturbed system, and blue represents the reachable set of the perturbed system.
    This analysis, using the reachable set computation method proposed in Section \ref{sec:setBased}, took 0.32 seconds.
    \item \textit{Answer to Question (4)}. Since we are interested in the impact of $k_{31}$ and $k_{21}$ on concentration level of only the first peripheral compartment ($c_1$), we projected our model to the $c_1$ dimension for this analysis. After the projection, we applied Algorithm \ref{algo:order} and found that the parameter $k_{31}$ has more impact than $k_{21}$: $k_{31}$ is more sensitive to perturbation than $k_{21}$ with regards to concentration level $c_1$. With an initial concentration level $c_p \in [2,4]$, $c_1 \in [4,6]$, $c_2 \in [4,6]$, $c_e \in [3,5]$ and $u \in [0,10]$, we separately studied the impact of $\pm$2\% on $k_{31}$ and $k_{21}$ --- perturbation in $k_{31}$ has more impact on the reachable set than perturbation in $k_{21}$. We illustrate this in Fig. \ref{fig:pkpd_topBot}: The $x$ and the $y$ axis represents the time steps and the concentration level of $c_1$ respectively. Cyan represents reachable sets with $\pm$2\% in $k_{21}$, and blue represents the reachable set with $\pm$2\% in $k_{31}$. This analysis, \textit{i.e.}, applying Algorithm \ref{algo:order} took just 0.024 seconds.
    \item \textit{Answer to Question (5)}. With an initial concentration level of $c_p \in [2,4]$, $c_1 \in [3,6]$, $c_2 \in [3,6]$, $c_e \in [2,4]$ and $u \in [2,10]$, we applied Heuristic \ref{heu:robustMet} to find out the amount of perturbation that can be introduced to the weight measurement of the patient without violating the safe concentration level of the plasma compartment $c_p$. In just 4.32 seconds, we found that an error of $\pm$0.4\% can be tolerated in the weight measurement without violating safety up-to 20 time steps.  The reachable sets, up-to 20 time steps and $\pm$0.4\% perturbation in weight, have been shown in Fig. \ref{fig:pkpd_robMet}: The $x$ and the $y$ axis represents the time steps and the concentration level $c_p$ respectively. Cyan represents reachable sets of the unperturbed system, and blue represents the reachable set of the perturbed system.
    
\end{itemize}

\begin{figure}
    \includegraphics[height=0.3\textheight]{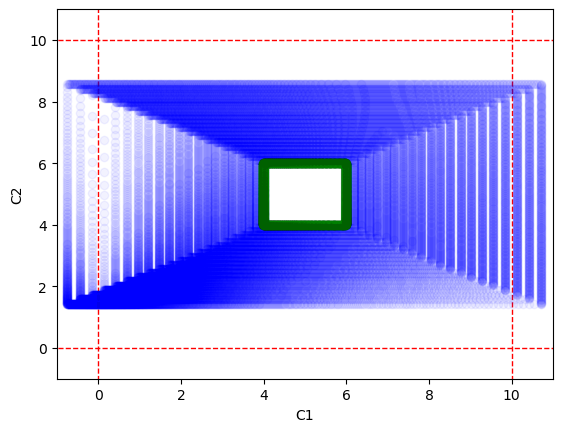}
    \caption{(Anaesthesia: Concentration level of $c_1$ and $c_2$ with perturbation in $k_{21}$ and $k_{31}$) The $x$ and the $y$ axis represents the concentration levels $c_1$ and $c_2$ respectively. Green: Phase plots of the reachable sets of the unperturbed system. Blue: Reachable sets of the perturbed system. Red: Lines demarcating the safe region. That is, $c_1$ is considered safe if it is between 0 to 10. The safety of $c_2$ is interpreted similarly.}
\label{fig:pkpd_c1c2}
\end{figure}

\begin{figure}
     \includegraphics[height=0.3\textheight]{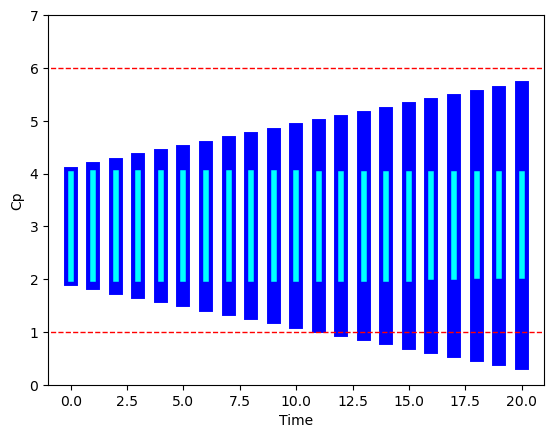}
    \caption{(Anaesthesia: Concentration level of $c_p$ with perturbation in the weight measurement of the patient)  The $x$ and the $y$ axis represents time steps and the concentration level $c_p$ respectively. Cyan: Phase plots of the reachable sets of the unperturbed system. Blue: Reachable sets of the perturbed system. Red: Lines demarcating the safe region. That is, $c_p$ is considered safe if it is between 1 to 6.}
\label{fig:pkpd_cp}
\end{figure}

\begin{figure}
     \includegraphics[height=0.3\textheight]{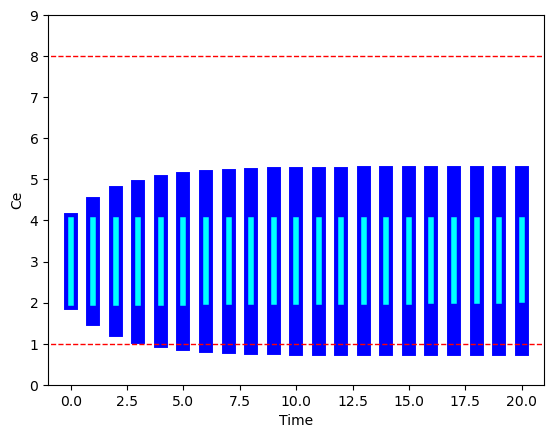}
    \caption{(Anaesthesia: Concentration Level of $c_e$ with perturbation in $k_d$) The $x$ and the $y$ axis represents time steps and the concentration level $c_e$ respectively. Cyan: Phase plots of the reachable sets of the unperturbed system. Blue: Reachable sets of the perturbed system. Red: Lines demarcating the safe region. That is, $c_e$ is considered safe if it is between 1 to 8.}
\label{fig:pkpd_ce}
\end{figure}

\begin{figure}
     \includegraphics[height=0.3\textheight]{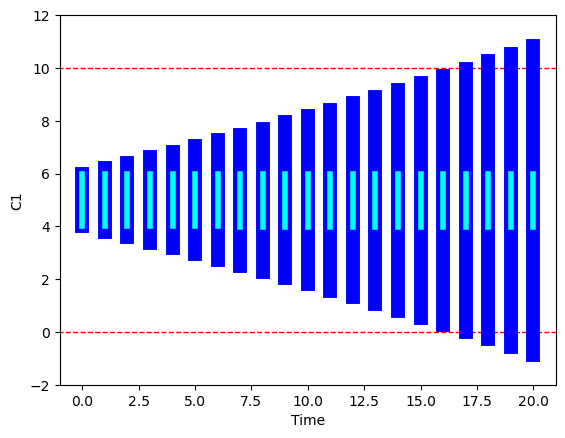}
    \caption{(Anaesthesia: Comparing concentration level of $c_1$ with perturbation in $k_{31}$ and $k_{21}$) The $x$ and the $y$ axis represents time steps and the concentration level $c_1$ respectively. Cyan: Phase plots of the reachable sets with perturbation of $\pm$2\% in $k_{21}$. Blue: Reachable sets with the same amount of perturbation in $k_{31}$. Red: Lines demarcating the safe region. That is, $c_1$ is considered safe if it is between 0 to 10.}
\label{fig:pkpd_topBot}
\end{figure}

\begin{figure}
     \includegraphics[height=0.3\textheight]{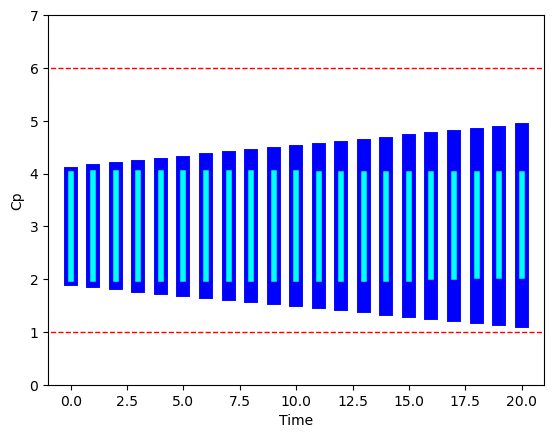}
    \caption{(Anaesthesia: Concentration level of $c_p$ with safe amount of perturbation in the weight measurement of the patient) The $x$ and the $y$ axis represents time steps and the concentration level $c_p$ respectively. Cyan: Phase plots of the reachable sets of the unperturbed system. Blue: Reachable sets with safe amount of perturbation in the weight measurement of the patient. Red: Lines demarcating the safe region. That is, $c_p$ is considered safe if it is between 1 to 6.}
\label{fig:pkpd_robMet}
\end{figure}

\subsection{Mars Rover}
\label{subsec:marsRover}
In this case study, similar to the previous one, we illustrate various artifacts proposed in this paper on a Mars Rover simulated on real Martian Terrain. In the robotics community, reachable set computation with model uncertainties has been previously considered to compute safety of a Mars Rover \cite{ghoshiros}. Here, in this case study, we not just limit our attention to determine safety of a Mars Rover's path, but also find out the parameters that are more sensitive to perturbations and also compute the amount of perturbation the rover can tolerate without violating safety --- to the best of the authors' knowledge, such questions have not been addressed before.

In this case study, a Mars Rover is required to determine the safety of its computed path, in a real Martian Terrain, where the model of the rover has uncertainties in its parameters like velocity, yaw angle, etc. We perform the safety analysis of a given path of the Mars Rover using the reachable set computation method proposed in Section \ref{sec:setBased}, with uncertainties in its model. We simulate the above mentioned scenario using NASA's HiRise Dataset \cite{hirise,nakanoya2020taskrelevant}, where the rover is required to safely maneuver in the Martian Terrain avoiding the steep obstacles. The rover is assumed to follow a linearized bicycle model with perturbations in the dynamics matrix. Similar to \cite{ghoshiros}, the path of the Mars Rover is computed as follows: given a set of waypoints (including the start and end coordinates) in the Martian Terrain, we compute a reference trajectory using Cubic Spline Planner \cite{pythonrobotics}, which is then used to plan the Mars Rover's control inputs using Model Predictive Control (MPC) \cite{pythonrobotics}.

\vspace{1ex}
\noindent \textbf{{\large Model of the Rover:}}
We assume a linearized bicycle model \cite{pythonrobotics} for our rover dynamics, with state variables: (i) $x$-position $x$, (ii) $y$-position $y$, (iii) velocity $v$, (iv) yaw angle $\phi$. The input to the rover dynamics are: (i) acceleration $a$, (ii) steering angle $\delta$. The states of the rover dynamics evolve by relating itself to other state variables through complicated trigonometric functions of $v$, $\phi$ and $\delta$. The detailed model can be found in \cite{pythonrobotics}.

Here, we not just limit our attention to compute reachable sets of the Mars Rover with model uncertainties, \textit{but also use the artifacts, developed in this paper, to identify  (rather) simpler trigonometric terms --- that are sensitive to perturbation --- relating it (the given state variable) to other state variables}. Specifically, we answer the following questions: 
\begin{enumerate}
    \item Impact of perturbation, in the yaw angle $\phi$, on the Mars Rover's path. 
    \item Identify simpler trigonometric terms --- that relate evolution of the $x$ and $y$ dimensions to other state variables, based on its sensitivity to perturbation. 
    \item Amount of perturbation that can be tolerated in the yaw angle $\phi$ without violating the safety of the Mars Rover. 
\end{enumerate}
The answers to the above questions will help the engineer/practitioner to carefully tune the error tolerances of the various sensors available. This is particularly useful in scenarios, such as this, where there is a limited resource available and therefore resources should be allocated very judiciously to various sensors and controllers.

\vspace{1ex}
\noindent \textbf{{\large Results:}} In the following experiments, to maintain consistency for comparing the results, we have used the same waypoints and therefore yielding the same MPC path. In the following figures, the red dots represent high raised obstacles (or unsuitable temperature) in the Martian Terrain.
\begin{itemize}
    \item \textit{Answer to Question (1)}. We compute the reachable sets, with $\pm$15\% perturbation in $\phi$, using the method proposed in Section \ref{sec:setBased}. In Fig. \ref{fig:mars_rs} we show the computed reachable sets: Blue represents the computed reachable sets with $\pm$15\% perturbation in $\phi$; the same reachable set is highlighted in red when it intersects with an unsafe point in the Martian Terrain. This analysis just took  2.30 seconds.
    \item \textit{Answer to Question (2)}. Since we are interested in identifying simpler terms that affects the evolution of just the $x$ and $y$ dimensions, in relation to perturbation, we projected our model to the $x$ and $y$ dimensions for this analysis. After the projection, we applied Algorithm \ref{algo:order} to find the simpler trigonometric terms ($i.e.$ elements of individual cells) that are most and least sensitive to perturbation with respect to $x-y$ dimension. The most sensitive terms for $x$ dimension are: $v\cdot sin(\phi) \cdot \phi$, and $cos(\phi)$ relating $x$ to the state variable $v$. The most sensitive terms for $y$ dimension are: $v\cdot cos(\phi)$ relating $y$ to state variable $\phi$, and $sin(\phi)$ relating $y$ to the state variable $v$. Similarly, the least sensitive terms for $x$ dimension are: $-v \cdot sin(\phi)$ relating $x$ to $\phi$, and $cos(\phi)$ relating $x$ to $v$.  The least sensitive term for $y$ dimension is: $-v \cdot cos(\phi) \cdot \phi$. By introducing a perturbation of $\pm$10\% to the most and the least sensitive terms, we compute the reachable sets separately, illustrating the impact of the top versus the bottom cells (as in this case terms), as returned by Algorithm \ref{algo:order}, on the reachable set. This is illustrated in Fig \ref{fig:mars_comp}: Blue represents reachable set with perturbation in the most sensitive cells, and green represents reachable set with perturbation in the least sensitive cells. Identifying the most and the least sensitive terms, using Algorithm \ref{algo:order}, took 0.04 seconds.
    \item \textit{Answer to Question (3)}. We applied Heuristic \ref{heu:robustMet} to find out the amount of perturbation that can be introduced to the yaw angle without violating the safety. In 2675.39 seconds, we found that an error of $\pm$9\% can be introduced to the yaw angle without violating safety.  The reachable sets with $\pm$9\% perturbation in the yaw angle have been shown in Fig. \ref{fig:mars_robMet}: Blue represents the computed reachable sets with $\pm$9\% perturbation in $\phi$.
\end{itemize}

\begin{figure}
     \includegraphics[height=0.4\textheight]{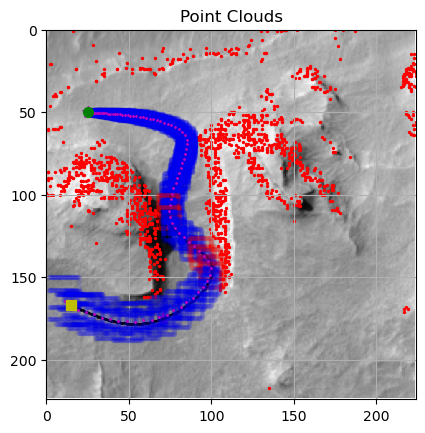}
    \caption{(Mars Rover: Reachable Set with perturbation in the yaw angle) Red Dots: Unsafe points in the Martian Terrain representing high raised obstacle or unsuitable temperature. Blue: Phase plots of the reachable sets with perturbation in the yaw angle. Red: The same reachable sets are highlighted in red when it intersects with an unsafe point in the Martian Terrain.}
\label{fig:mars_rs}
\end{figure}

\begin{figure}
    \includegraphics[height=0.4\textheight]{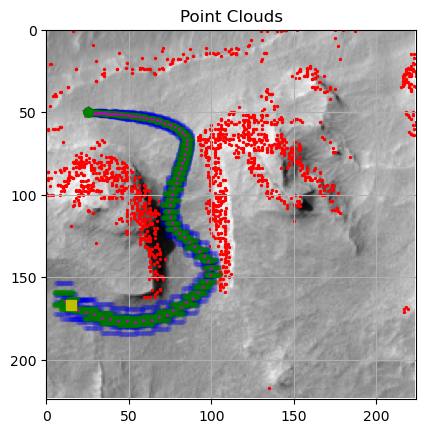}
    \caption{(Mars Rover: Comparing reachable sets with perturbation in the most and the least sensitive terms) Red Dots: Unsafe points in the Martian Terrain representing high raised obstacle or unsuitable temperature. Blue: Phase plots of the reachable sets with perturbation in the most sensitive terms. Green: Reachable sets with the same amount of perturbation in the least sensitive terms.}
\label{fig:mars_comp}
\end{figure}

\begin{figure}
    \includegraphics[height=0.4\textheight]{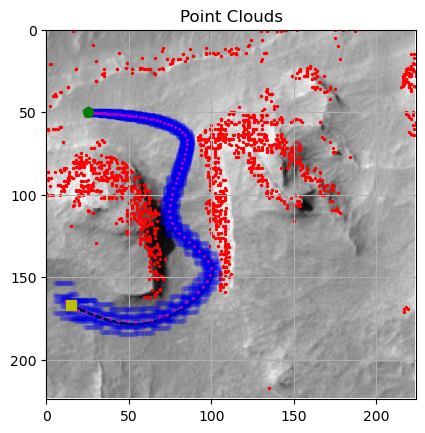}
    \caption{(Mars Rover: Reachable Set with tolerable perturbation in the yaw angle) Red Dots: Unsafe points in the Martian Terrain representing high raised obstacle or unsuitable temperature. Blue: Phase plots of the reachable sets with tolerable perturbation in the yaw angle.}
\label{fig:mars_robMet}
\end{figure}

\section{Evaluation on Benchmarks}
\label{sec:evaluation}

In Section \ref{sec:caseStudies}, we demonstrated the applicability (as well as scalability) of the artifacts, proposed in this paper, on two real-life applications, taken from two different domains --- suggesting the broad-spectrum applicability of our approach. In this section, on the other hand, we further evaluate our methods on several linear system benchmarks.

\vspace{1ex}
\noindent \textbf{{\large Methodology:}} 
We evaluate our methods on a set of linear dynamics benchmarks from the proceedings of ARCH workshop \cite{ARCH}.
The initial set is selected as per the benchmark description and reachable set computation is performed for 20 time units. 
We introduce perturbations in the dynamics based on the description of the benchmark; if the benchmark does not have any inherent perturbations or parameters, we introduce perturbations at random cells. 

To illustrate the symbolic reachable set computation technique, we compute the bloating factors using \textbf{\texttt{Kagstrom1}}, \textbf{\texttt{Kagstrom2}}, and \textbf{\texttt{Loan}} as described in Section~\ref{sec:perturbation}. 
For illustrating the set-based (numeric) technique, as described in Section \ref{sec:setBased}, we compute the discrete linear uncertain system with step size of $0.01$ and compute: (a) reachable set without generator reduction (Equation~\ref{eq:ORScomp}), (b) with generator reduction as per the first heuristic (Interval Arithmetic Based Reduction) at every $500^{th}$ step, and (c) generator reduction using \texttt{pypolycontain} every $500^{th}$ step.

To illustrate the effect of singular values and reachable set, we introduce uncertainties in the five most sensitive and five least sensitive elements in the matrix that change the maximum singular value.
That is, we order the indices as provided by Algorithm~\ref{algo:order} and introduce uncertainties in the first five and the last five cells.
%
%

Finally, we estimate the robustness threshold of the benchmarks using $\tnm{Proportional}$ $ \tnm{Distribution (Prop. Dist.)}$ and $\tnm{Harmonic Distribution (Har. Dist.)}$ heuristics and compare them to the base line of $\tnm{Equal Distribution}$ $\tnm{(Eq. Dist.)}$.
That is, we add uncertainties to the benchmarks according to each of the heuristics and terminate when any increase in uncertainty would violate the safety property.
%

\begin{table}[t]
    \centering
    \resizebox{400pt}{!}{
    \begin{tabular}{| c | c | c | c | c | c | c | c | c | c |}
    \hline
~Benchmark~ & ~Dim~ & ~ Pert ~ & ~w/o Reduction~ & ~Tool~ &  ~ Reduction~ &  ~ \tnm{Kagstrom1} ~ &  ~ \tnm{Kagstrom2} ~ &  ~ \tnm{Loan} ~ &  ~ Ordering ~\\ \hline
Holes & 10 & 4 & 11583 s & 3228 s & 2814 s & 0.004 s & 0.4 s & 0.004 s & 0.01s \\ \hline
ACC & 4 & 2 & 1144 s & 301 s & 266 s &  0.02 s & 0.06 s & 0.014 s & 0.0006 s\\
\hline
Lane Change & 7 & 3 & 3260 s & 966 s & 817 s & 0.007 s & 0.29 s & 0.001 s & 0.004 s\\
\hline
PK/PD & 4 & 3 & 1941 s & 510 s & 447 s & 0.005 s & 9.06 s & 0.0012 s & 0.001 s\\
\hline
Motor & 7 & 2 & 2952 s & 906 s & 747 s & 0.06 s &  0.11 s & 0.008 s & 0.004 s\\
\hline
Girad I & 2 & 2 & 395 s & 100 s & 93 s & 0.002 s & 0.02 s & 0.001 s & 0.0003 s\\
\hline
Girad II & 5 & 3 & 2503 s & 635 s & 580 s & 0.002 s & 0.28 s & 0.0006 s & 0.005 s\\
\hline
Space & 6 & 4 & 3370 s & 931 s & 819 s & 0.002 s & 0.2 s & 0.0006 s & 0.002 s\\
\hline
Aircraft & 4 & 2 & 1438 s & 368 s & 331 s & 0.01 s & 0.04 s & 0.01 s & 0.0008 s\\
\hline
CoOp I & 10 & 6 & 11020 s & 3461 s & 2716 s & 0.004 s & -  & 0.0008 s & 0.01 s\\
\hline
Flight & 16 & 2 & 19546 s & 6324 s & 4706 s & 0.03 s & 2.009 s & 0.002 s & 0.2 s\\
\hline
5-Veh & 15 & 6 & 31742 s & 9190 s & 7620 s & 0.002 s & -  & 0.0006 s & 0.4 s\\
 \hline
    \end{tabular}%
    }
    \caption{ \textbf{Reachable Set Computation.} (Time taken by the numerical approach is for 2050 steps). \textbf{Dim}: Dimensions of the system, \textbf{Pert}: Number of cells perturbed, \textbf{w/o Reduction}: Time taken by the numerical method without any generator reduction, \textbf{Tool}: Time taken by numerical method using \texttt{pypolycontain} for generator reduction, \textbf{Reduction}: Time taken by the numerical method with generator reduction. \textbf{\texttt{Kagstrom1}}: Time taken by \tnm{Kagstrom1}, \textbf{\texttt{Kagstrom2}}: Time taken by \tnm{Kagstrom2}, \textbf{\texttt{Loan}}: Time taken by Loan. \textbf{Ordering}: Time taken by Algorithm \ref{algo:order} to order the cells of the matrix based on sensitivity to perturbation.}
    \label{tab:results1}
\end{table}

\begin{table}[t]
    \centering
    \resizebox{300pt}{!}{
    \begin{tabular}{| c | c | c | c | c | c | c | c | c | c | c |}
    \hline
~Benchmark~ & ~Dim~ & Nom & Prop ~ & Har & Equal & Time\\ \hline
Holes & 10 & 3.079 & 24861.27 & 15434.37 & 2684.47 & 4127.08 s \\ \hline
ACC & 4 & 2.000 & 0.0966 & 0.0966 & 0.10 & 95.93 s\\
\hline
Lane Change & 7 & 2.613 & 0.020 & 0.01 & 0.0091 & 260.14 s\\
\hline
PK/PD & 4 & 2.124 & 0.0037 & 0.003 & 0.0025 & 104.93 s\\
\hline
Girad I & 2 & 1.401 & 0.209 & 0.209 & 0.16 & 0.87 s\\
\hline
Girad II & 5 & 2.192 & 0.092 & 0.09 & 0.09 & 206.71 s\\
\hline
Space & 6 & 2.455 & 0.16 & 0.208 & 0.028 & 566.91 s\\
\hline
Aircraft & 4 & 2.000 & 0.07 & 0.07 & 0.027 & 157.02 s\\
\hline
CoOp I & 10 & 3.132 & 3.94e-7 & 4.01e-7 & 1.66e-6 & 72.943 s\\
\hline
Flight & 16 & 4.002 & 6.92e-7 & 6.92e-7 & 6.92e-7 & 124.86s\\
\hline
5-Veh & 15 & 3.821 & 22.16 & 22.75 & 21.75 & 2734.63 s\\
 \hline
    \end{tabular}%
    }
    \caption{ \textbf{Robustness Threshold (Algorithm \ref{heu:robustMet})}. \textbf{Dim}: Dimensions of the system, \textbf{Nom}: Frobenius norm of the nominal dynamics, \textbf{Prop}: Frobenius norm of the total perturbation introduced using the heuristic \tnm{Prop Dist.}, \textbf{Har}: Frobenius norm of the total perturbation introduced using the heuristic \tnm{Har. Dist.}, \textbf{Equal}: Frobenius norm of the total perturbation introduced using the heuristic \tnm{Eq. Dist.}, \textbf{Time}: Time taken by the heuristic.}
    \label{table:robMet}
\end{table}

\vspace{1ex}
\noindent 
\textbf{{\large Results:}} We divide the results into three main parts. The first part compares the efficiency and the accuracy of the reachable sets computed from symbolic and numerical approaches. The second part illustrates the effect of introducing uncertainties in most and least sensitive parts of the dynamics. The third part estimates the robustness threshold using the various heuristics.

\vspace{0.5ex}
\noindent
\textbf{Reachable Set Computation:} The time taken to compute the reachable set using the symbolic and numerical methods is contained in Table~\ref{tab:results1}.
Since the symbolic approaches just require evaluating closed form expressions, they take very less time to compute the overapproximation of reachable set (all of them run is less than 2 seconds). 
%
The bloating factors for the symbolic methods for ACC \cite{7349170} (more details in Subsection \ref{subsec:acc}) benchmark has been plotted in Figure~\ref{fig:ACC} (Left). The bloating factors for Lane Change benchmark is shown in Figure~\ref{fig:LaneChange} (Left). Note that, in ACC benchmark, performance of \textbf{\texttt{Kagstrom2}} and \textbf{\texttt{Loan}} are same.

The time taken for computing the reachable set using the generalized star representation and various optimizations is also provided in Table~\ref{tab:results1}.
Phase plots of the reachable set with no optimization (Red), interval arithmetic based generator reduction (Green), \texttt{pypolycontain} (Cyan), reachable set without perturbation (Blue), and sample trajectories from initial set with perturbations generated randomly from $\Lambda$ (Purple), for the ACC benchmark is provided in Figure~\ref{fig:ACC} (Right).
While the time taken for computation grows with the time horizon, the reachable set is very accurate, as illustrated by the reachable set of the various sample dynamics.

\vspace{0.5ex}
\noindent
\textbf{Perturbations and Singular Values:} While we have proved in Theorem~\ref{thm:2x} that the effect of perturbation on the reachable set is proportional to the maximum singular value, we illustrate the effect of introducing uncertainties in the five highest sensitive and five least sensitive cells in the dynamics.
That is, we pick the first $5$ and the last $5$ cells in the $\mathit{OrderedList}$ computed by the Algorithm~\ref{algo:order}. The comparison of the two reachable sets are illustrated, for two of the benchmarks, in Figures~\ref{fig:GiradIOrd} and~\ref{fig:fiveOrdAppx}. It is clear that changing the elements that affect the maximum singular values has a disproportionate effect on the reachable sets. While we couldn't include illustration from all other benchmarks, this behavior is observed in all benchmarks.
%

\begin{figure}
    \includegraphics[height=0.25\textheight]{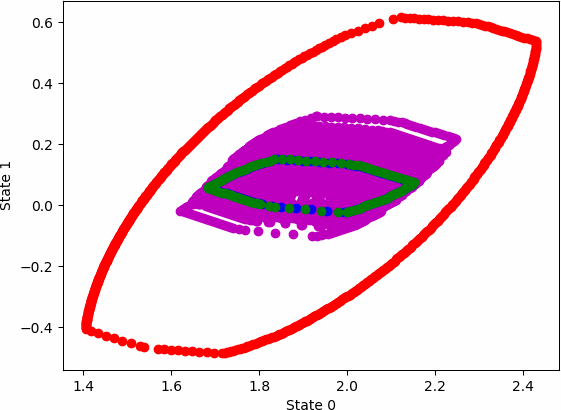}
    \caption{(Effect of perturbation on various cells, as ordered by Algorithm \ref{algo:order}, on the benchmark Five Vehicle Platoon) Red: Reachable set after perturbing the top 5 cells as returned by Algorithm \ref{algo:order}. Green: Reachable set after perturbing the bottom 5 cells as returned by Algorithm \ref{algo:order}. Blue: Reachable set with no perturbation. Purple: Sample trajectories from initial set with perturbations generated randomly from $\Lambda$}
\label{fig:fiveOrdAppx}
\end{figure}

\vspace{0.5ex}
\noindent
\textbf{Estimation of Robustness Threshold:} The robustness threshold estimated using the heuristics proposed in Section~\ref{sec:uncertainties} is provided in Table~\ref{table:robMet}.
From the description of Algorithm~\ref{heu:robustMet}, we get that the robustness of the system would depend on the unsafe set chosen.
In our case, we pick the unsafe set very close to the trajectories starting from the initial set (at a constant distance from the trajectories near the initial set). 
The projections of reachable set for the same uncertainty distributed according to the two heuristics \tnm{Prop. Dist.} and \tnm{Har. Dist.} on Lane Change Maneuvers for Autonomous Vehicles \cite{7553499} and Girad I \cite{10.1007/978-3-540-31954-2_19} benchmarks are provided in Figure \ref{fig:heu_distro} (details provided in Subsections \ref{subsec:laneChange} and \ref{subsec:other_benchmarks}).
In almost all the benchmarks, \tnm{Prop. Dist.} and \tnm{Har. Dist.} heuristics would discover uncertainties with larger magnitude that are safe when compared to \tnm{Eq. Dist.}. 
Notice from Table~\ref{table:robMet} that except for the cases that are very robustly safe (Holes, 5-Veh, etc), we could only introduce perturbations with Frobenius norm roughly $1-10\%$ of the system.
This further illustrates that for estimating the robustness of safety of a linear system with respect to model uncertainties, for performing adversarial analysis, one should focus on introducing uncertainties that increase the maximum singular values.

%

\subsection{Benchmarks}
In the rest of this section, we provide the details of the various benchmarks on which our algorithms and heuristics were evaluated.

\begin{figure}[t]
\centering
\includegraphics[height=0.25\textheight]{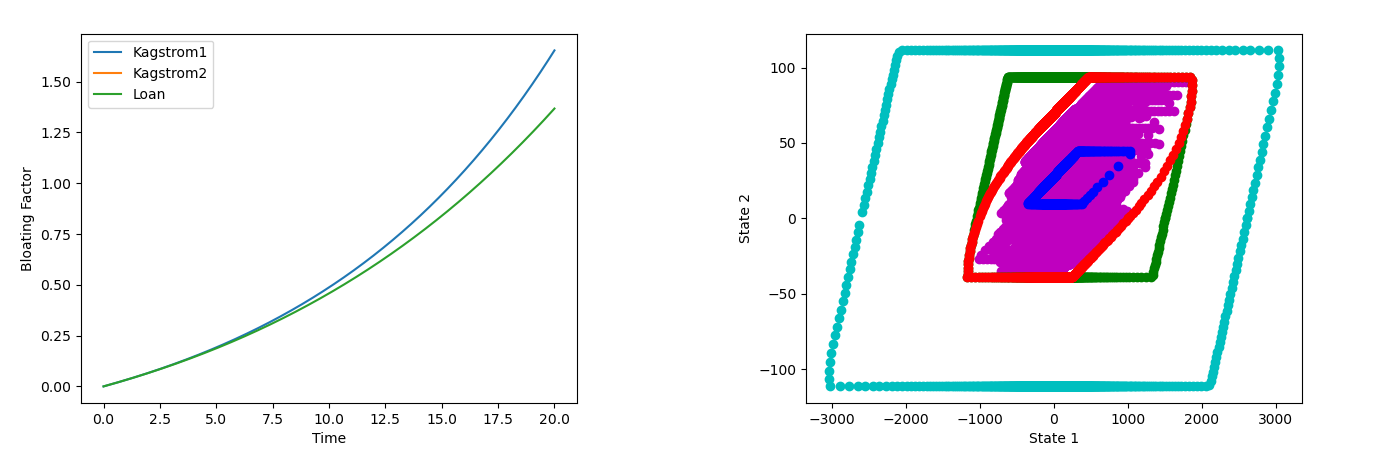}
\caption{\textbf{Left}: (Bloating Factors of ACC) The computed bloating factors (\tnm{Kagstrom1}, \tnm{Kagstrom2}, \tnm{Loan}) from time 0 to 20 with step size of 0.01. \textbf{Right} (Reachable sets of ACC at time step 1950) The reachable sets at time step 1950 returned by our numerical approach is visualized (States $v$, $h$). Red: Phase plots of the reachable set with no reduction, Green: Interval arithmetic based generator reduction, Cyan: \texttt{pypolycontain} reduction, Blue: Reachable set without perturbation, Purple: Sample trajectories from initial set with perturbations generated randomly from $\Lambda$}
\label{fig:ACC}
\end{figure}

\subsubsection{Adaptive Cruise Controller}
\label{subsec:acc}
In \cite{7349170}, Adaptive Cruise Control (ACC) behaves like a ordinary Cruise Control when there is no preceding vehicle in the sight of its sensor --- it maintains a constant speed set by the driver. When a preceding car is detected by the sensor, the car automatically adjusts its speed and maintains a safe distance with the preceding vehicle.
\emph{Model:} The ACC vehicle is modeled as a (lumped) point mass $m$ moving along a straight line with velocity $v$. And braking and engine torque applied to the wheels is lumped as a net force. The dynamics is given in equation 6(a), 6(b) and 6(c) of \cite{7349170}. The state variables are: (i) $v$: Velocity of the vehicle, (ii) $h$: distance from the lead vehicle, $v_L$: velocity of the lead vehicle. The dynamics is dependent on parameters like: $F$, braking and engine torque applied to the wheels is lumped as a net force; $a_L$, acceleration of the lead vehicle.
We introduce perturbation in the acceleration of lead vehicle $a_L$ and $\bar{F}_w$. We chose our initial set as [0,35] $\times$ [5,50] $\times$ [0,35] $\times$ [1,1]. The reachable sets at time step 1950 returned by our numerical approach is visualized (States $v$, $h$) in Figure ~\ref{fig:ACC} (Right). Using our symbolic approach, we computed the bloating factors from time 0 to 20 with step size of 0.01. The result of the bloating factor computation are shown in Figure \ref{fig:ACC} (Left). The ordering of the top 5 cells as returned by our approach, with decreasing sensitivity to perturbation are \{(2, 2), (0, 0), (3, 3), (1, 1), (1, 2)\}. The computation time taken by all the above mentioned methods is given in Table \ref{tab:results1}, \ref{table:robMet}. The reachable set computation time, for the numeric approach, is given for 2050 steps.

\subsubsection{Five Vehicle Platoon}
\label{subsec:5veh}
In this subsection we compute reachable sets of a 15 dimensional, 5 vehicle platoon model \cite{vehPlatoon}. 
This benchmark is a framework of 5 autonomously driven vehicles --- one of the vehicle is a leader, located at the head of the formation, and the rest of the vehicles act as a follower.
The vehicles establish synchronization through communication via network, exchanging information about their relative positions, relative velocities, accelerations measured with on-board sensors.
Each vehicle in the platoon is described by a 3 dimensional state vector --- relative position, its derivative and acceleration.
The main goal is to avoid collisions inside the platoon.
\emph{Model}: For a vehicle $i$: $e_i$ is its relative position, $\dot{e_i}$ is the derivative of $e_i$, and $a_i$ is the acceleration. The state variables are: [$e_1$, $\dot{e_1}$, $a_1$, $e_2$, $\dot{e_2}$, $a_2$, $e_3$, $\dot{e_3}$, $a_3$, $e_4$, $\dot{e_4}$, $a_4$, $e_5$, $\dot{e_5}$, $a_5$]. The dynamics is provided in \cite{vehPlatoon}.
In Five Vehicle Platoon with initial set $[0.9,1.1]^{15}$, we perform the following steps:
(i) plot the reachable set of unperturbed system (represented in \texttt{Blue});
(ii) introduce perturbation of 1\% in the bottom five cells as returned by Algorithm~\ref{algo:order}. Then compute the over-approximate reachable set using our numerical method. (represented in \texttt{Green});
(iii) introduce perturbation of 1\% in the top five cells as returned by Algorithm~\ref{algo:order}. Then plot the reachable set of 50 random samples within the perturbation range. (represented in \texttt{Magenta});
(iv) introduce perturbation of 1\% in the top five cells as returned by Algorithm~\ref{algo:order}. Then compute the over-approximate reachable set using our numerical method. (represented in \texttt{Red}). The reachable sets are shown in Figure \ref{fig:fiveOrdAppx}. In this is example, there is no visible effect of perturbation on the bottom 5 cells. The computation time taken by all the above mentioned methods is given in Table \ref{tab:results1}, \ref{table:robMet}. The reachable set computation time, for the numeric approach, is given for 2050 steps.

\begin{figure}
\centering
\includegraphics[height=0.25\textheight]{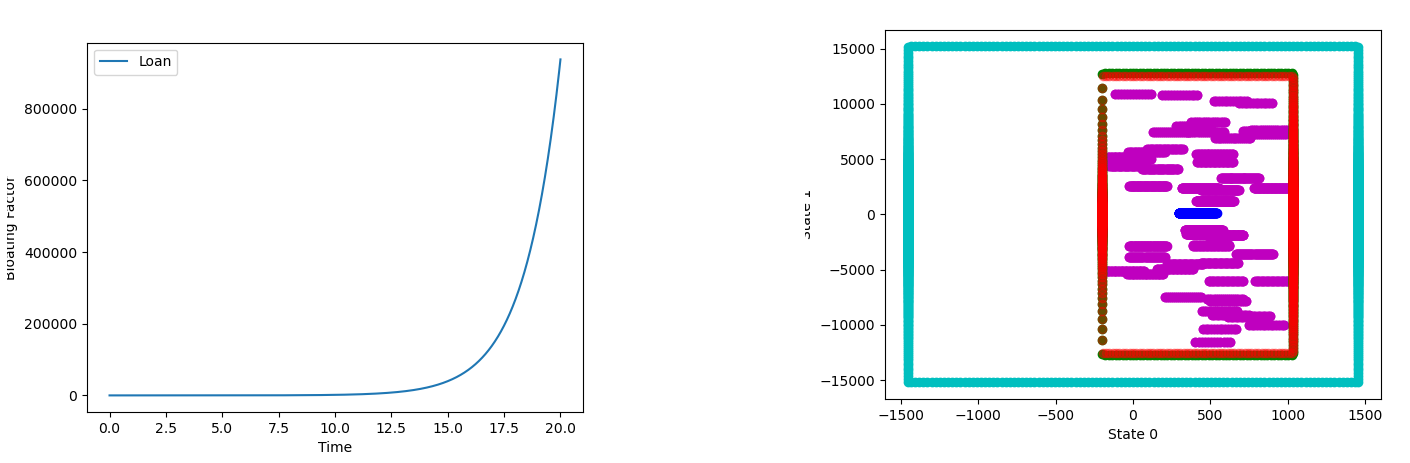}
\caption{\textbf{Left}: (Bloating Factors of Lane Change Maneuver for Autonomous Vehicles) The computed bloating factors (\tnm{Kagstrom1}, \tnm{Kagstrom2}, \tnm{Loan}) from time 0 to 20 with step size of 0.01. \textbf{Right} (Reachable sets of Lane Change Maneuver for Autonomous Vehicles at time step 1950) The reachable sets at time step 1950 returned by our numerical approach is visualized (States $x_r$, $y_r$). The color coding is same as Fig. \ref{fig:ACC} (Right).} 
\label{fig:LaneChange}
\end{figure}

\begin{figure}[t]
\centering
\includegraphics[height=0.24\textheight]{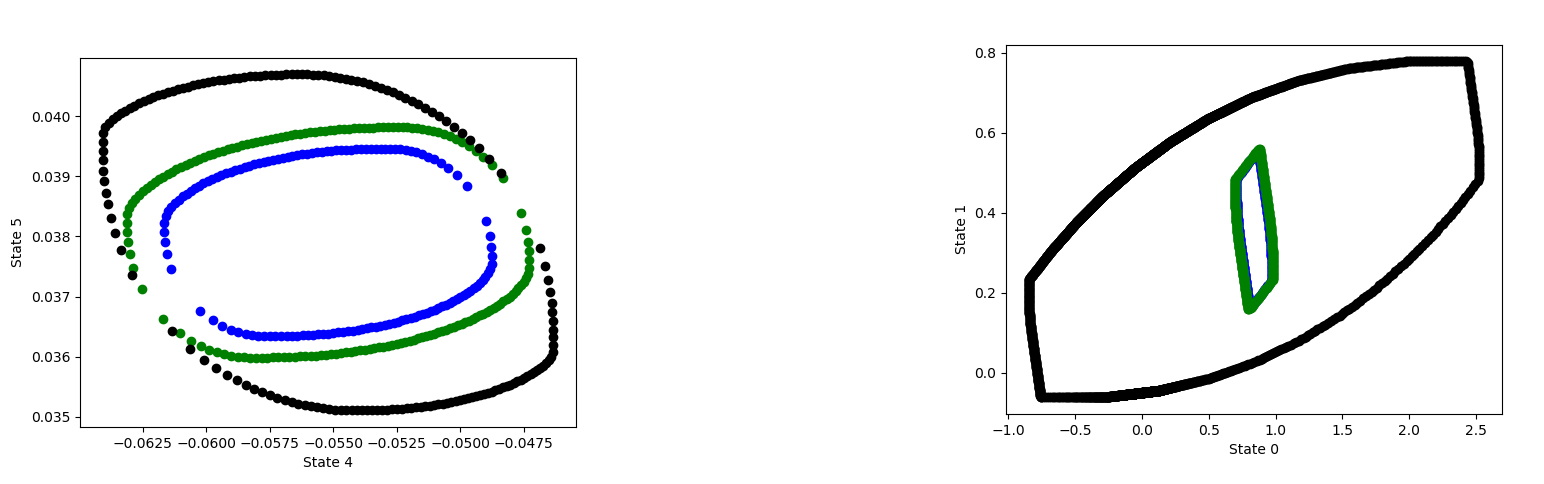}
\caption{\textbf{Left}: (Effects of distributing perturbation on the Lane Change Maneuvers for Autonomous Vehicles benchmark) The reachable set in blue represents a budget is distributed according to \tnm{Har. Dist.}, the reachable set in green represents a the same budget is distributed according to \tnm{Prop. Dist.}, and the reachable set in black represents a the same budget is distributed according to \tnm{Eq. Dist.}. \textbf{Right} (Effects of distributing perturbation on the Girad I benchmark) Distributions as per \tnm{Prop. Dist.}, \tnm{Har. Dist.} and \tnm{Eq. Dist.}. The color coding is same as the left figure.}
\label{fig:heu_distro}
\end{figure}

\subsubsection{Lane Change Maneuver for Autonomous Vehicles}
\label{subsec:laneChange}
In \cite{7553499}, the authors provide a 7 dimensional lane change maneuver model for autonomous vehicles. An autonomous vehicle is an extremely safety critical system, and lane change is one of its key aspects. According to the National Highway Traffic Safety
Administration, the main cause of lane change accidents is failure to detect the other vehicle and almost 80\% of the accidents occur at speeds smaller than 25km/h. \cite{7553499} provides a benchmark, modeling a cooperative lane change maneuver that involves four autonomous cars --- one in the right lane and three in the left lane. The vehicle in the
right lane wants to make a lane change and moving behind the second vehicle and ahead of the third one. It does so, while (i) maintaining safety margins with all the surrounding vehicles, (ii) respecting the traffic rules, and (iii) satisfying
the physical and design limitations.
\emph{Vehicle Dynamics}: This benchmark considers a dynamic bicycle model with a linear tire model. The selection of a dynamic model is made as the tire slip describes important phenomena, e.g. under-steering. At the same time, the model is assumed linear to avoid computational complexity. The model dynamics is given in equation 6 of \cite{7553499}. The state variables are: (i) $x_r$: longitudinal position of the rear axle, (ii) $y_r$: the lateral position of the rear axle, (iii) $\psi$: the yaw angle, (iv) $v_x$: the longitudinal velocity, (v) $v_y$: the lateral velocity at the center of the rear axle, (vi) $\omega$: the yaw rate.
We introduce perturbation in cells \{(3,6), (4,6), (5,6)\} according to the input parameters. We chose our initial set as [0,50] $\times$ [3,3.5] $\times$ [0,0] $\times$ [20,30] $\times$ [0,0] $\times$ [0,0] $\times$ [1,1]. The reachable sets at time step 1950 returned by our numerical approach is visualized (States $x_r$, $y_r$) in Figure~\ref{fig:LaneChange} (Right). Using our symbolic approach, we computed the bloating factors from time 0 to 20 with step size of 0.01. The result of the bloating factor computation are shown in Figure \ref{fig:LaneChange} (Left)  --- here, only best performance has been plotted, since performance of the other two bounds were distorting the plot scale. The ordering of the top 5 cells as returned by our approach, with decreasing sensitivity to perturbation are \{(3, 3), (6, 6), (1, 1), (2, 2), (0, 0)\}. The effect of distributing a perturbation of 20\% on cells \{(3,6), (4,6), (5,6)\} is given in Figure \ref{fig:heu_distro} (Left). The reachable sets in Figure \ref{fig:heu_distro} (Left) are for time step 50. The reachable set in blue represents the budget is distributed according to \tnm{Har. Dist.}, the reachable set in green represents a the same budget is distributed according to \tnm{Prop. Dist.}, and the reachable set in black represents a the same budget is distributed according to \tnm{Eq. Dist.}. The computation time taken by all the above mentioned methods is given in Table \ref{tab:results1}, \ref{table:robMet}. The reachable set computation time, for the numeric approach, is given for 2050 steps.

\begin{figure}
\centering
\includegraphics[height=0.24\textheight]{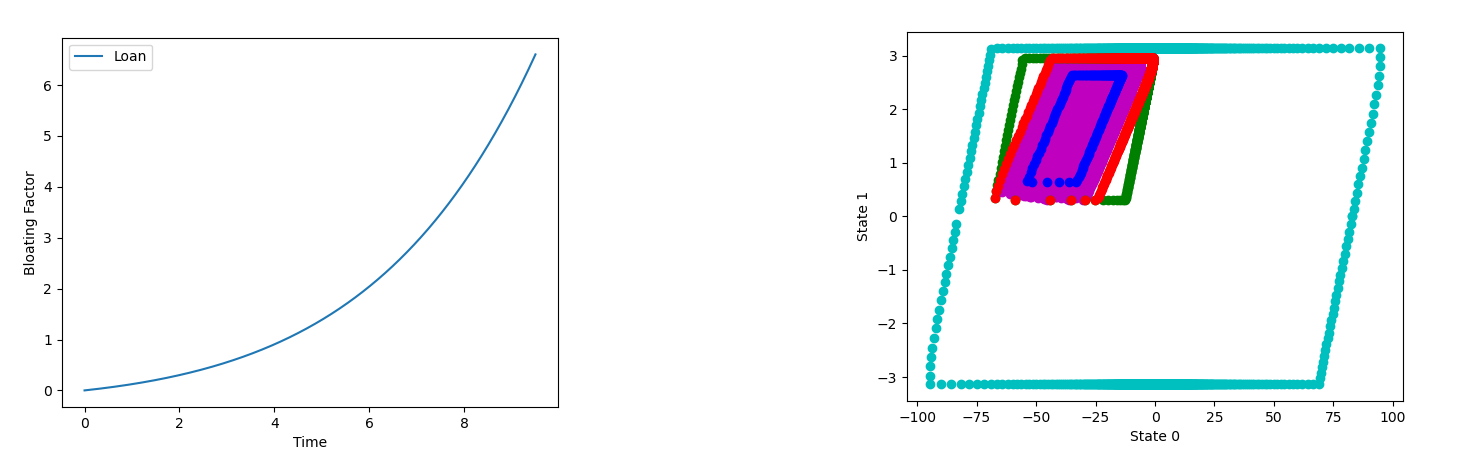}
\caption{\textbf{Left}: (Bloating Factors of Self-Balancing Two-wheeled Robot) The computed bloating factors (\tnm{Kagstrom1}, \tnm{Kagstrom2}, \tnm{Loan}) from time 0 to 20 with step size of 0.01. \textbf{Right} (Reachable sets of Self-Balancing Two-wheeled Robot at time step 950) The reachable sets at time step 950 returned by our numerical approach is visualized (States $\theta_{err}$, $\theta_{ref}$). The color coding is same as Fig. \ref{fig:ACC}.}\label{fig:Holes}
\end{figure}


\subsubsection{Self-Balancing Two-wheeled Robot}
\label{subsec:2robot}
In \cite{DBLP:conf/cpsweek/HeinzOW14}, the authors provide a 10 dimensional benchmark of an open source robotics application --- a self-balancing two-wheeled robot, modelled as an inverted pendulum. The plant is a two-wheeled robot where both wheels are driven by independently controlled servo motors. The plant relates motor voltage of the left and right servo motor. And the controller responsible for balancing the robot, provides two input signals that are used to control transnational and rotational motion that can be driven by any path planning application. If these algorithms yield a state from which the balance controller is unable to recover, then the robot falls over.
\emph{Model}: The state vector of the plant is composed of the average angle of left and right wheel $\theta$, body pitch angle $\psi$, body yaw angle $\phi$, and the respective angular velocities.
The actual input of the plant is an encoding as pulse-width modulation
duty cycles. The state variables are given in Section 2.1 of \cite{DBLP:conf/cpsweek/HeinzOW14}. And its control goal is to steer the transnational speed towards the reference value $\theta_{ref}$ and to steer the body pitch angle $\psi$ towards 0. Given these, the state variables of the controller is given in Section 2.2 of \cite{DBLP:conf/cpsweek/HeinzOW14}. And the controller dynamics is given in appendix of \cite{DBLP:conf/cpsweek/HeinzOW14}.
We introduced perturbation of $\pm20$\% in cells \{(0,3), (1,2), (3,2), (4,3)\}. We chose our initial set as [-1,1] $\times$ [-1,1] $\times$ [1,1] $\times$ [-1,1] $\times$ [-1,1] $\times$ [1,1] $\times$ [1,1] $\times$ [1,1] $\times$ [1,1] $\times$ [1,1]. The reachable sets at time step 950 returned by our numerical approach is visualized (States $\theta_{err}$, $\theta_{ref}$) in Figure~\ref{fig:Holes} (Right). Using our symbolic approach, we computed the bloating factors from time 0 to 20 with step size of 0.01. The result of the bloating factor computation is in Figure \ref{fig:Holes} (Left) --- here, only best performance has been plotted, since performance of the other two bounds were distorting the plot scale. The ordering of the top 5 cells as returned by our approach, with decreasing sensitivity to perturbation are \{(5, 5), (3, 3), (9, 9), (0, 0), (6, 6)\}.
%
%
The effect of distributing a fixed perturbation of on all the cells of the dynamics matrix is given in Figure \ref{fig:heu_distro} (Right). The reachable sets in Figure \ref{fig:heu_distro} (Right) are for time step 10. The reachable set in blue represents the budget is distributed according to \tnm{Har. Dist.}, the reachable set in green represents a the same budget is distributed according to \tnm{Prop. Dist.}, and the reachable set in black represents a the same budget is distributed according to \tnm{Eq. Dist.}. The computation time taken by all the above mentioned methods is given in Table \ref{tab:results1}, \ref{table:robMet}. The reachable set computation time, for the numeric approach, is given for 2050 steps.

\subsubsection{Other Benchmarks}
\label{subsec:other_benchmarks}
In \cite{10.1007/978-3-540-31954-2_19}, Girad I, the effect of perturbation on various cells, as ordered by Algorithm \ref{algo:order}, at time step 200 with initial set $[0.9,1.1] \times [-0.1,0.1]$ is shown in Figure \ref{fig:GiradIOrd}. The computation time taken by all the above mentioned methods is given in Table \ref{tab:results1}.
We investigated several other benchmarks:  \cite{10.1007/978-3-540-31954-2_19} 
(Girad II), \cite{chan2017verifying} (Spacecraft), \cite{DBLP:conf/emsoft/LalP15} (Aircraft Dynamics), \cite{ARCH15:Networked_Cooperative_Platoon_of} (CoOpI), \cite{inproceedings} (Flight), Five Vehicle Platoon (5-Veh) \cite{vehPlatoon}. We are not able to provide the plots for these benchmarks due to lack of space. The computation time taken by all the above mentioned methods is given in Table \ref{tab:results1}, \ref{table:robMet}. The reachable set computation time, for the numeric approach, is given for 2050 steps.

\begin{figure}
   \includegraphics[height=0.25\textheight]{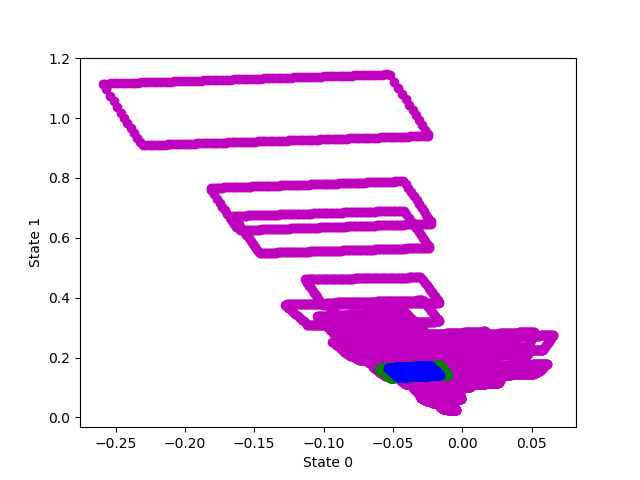}
    \caption{(Effect of perturbation on various cells, as ordered by Algorithm \ref{algo:order}, on the benchmark Girad I) Magenta: Sample reachable sets after perturbing the top 5 cells as returned by Algorithm \ref{algo:order}. Green: Reachable set after perturbing the bottom 5 cells as returned by Algorithm \ref{algo:order}. Blue: Reachable set with no perturbation.}
\label{fig:GiradIOrd}
\end{figure}

\section{Conclusions and Future Work}
\label{sec:conclusion}
In this paper, we proposed two approaches to compute reachable sets of linear uncertain systems, numeric and symbolic. 
This is the first work which extends the perturbation based techniques to interval uncertainties. 
On one hand, the proposed symbolic approach is computationally very efficient, but it is very conservative. 
On the other hand, the numeric approach, though computationally expensive, is more accurate. 
%

Finally, we provide more insights into the effects of uncertainties on the reachable set. 
We establish a relationship between the reachable set and the maximum singular value of the uncertain dynamics. 
Using this relation, we propose an algorithm to order the cells of the dynamics matrix based on its sensitivity to perturbation \textit{vis-à-vis} reachable sets. 
Leveraging the ordering algorithm, we propose and evaluate two new heuristics to compute robustness threshold of a linear dynamic system.

\bibliographystyle{ACM-Reference-Format}
\bibliography{main}


\appendix

\end{document}